\documentclass[aps,pra,onecolumn,nobibnotes,superscriptaddress,nofootinbib]{revtex4-2}
\usepackage{graphicx} % Required for inserting images
\usepackage{caption,graphics,physics}
\usepackage{subcaption}
\usepackage{amsmath,xcolor,bm}
\usepackage{qcircuit}
\usepackage{mathtools,amsmath,amsthm,amsfonts,amssymb,amscd}
\usepackage{hyperref}
\usepackage[ruled,vlined]{algorithm2e}
\usepackage{algpseudocode}

\usepackage[capitalize]{cleveref}

\hypersetup{
    colorlinks=true,
    linkcolor=red,
    filecolor=magenta,      
    urlcolor=cyan,
}
\usepackage{geometry}

\newtheorem{theorem}{Theorem}

\newtheorem{lemma}{Lemma}

\newtheorem{example}{Example}

\newtheorem{proposition}{Proposition}
\newtheorem{assumption}{Assumption}
\newtheorem{definition}{Definition}

\newcommand{\bs}{\boldsymbol}

\Crefname{appsec}{Appendix}{Appendices}

\begin{document}

\title{Classical optimization algorithms for diagonalizing quantum Hamiltonians}

\author{Taehee Ko}
\email{kthmomo@kias.re.kr}
\affiliation{School of Computational Sciences, Korea Institute for Advanced Study}

\author{Sangkook Choi}%
\affiliation{School of Computational Sciences, Korea Institute for Advanced Study}

\author{Hyowon Park}
\affiliation{Department of Physics, University of Illinois at Chicago}

\author{Xiantao Li}
\email{xxl12@psu.edu}
\affiliation{Department of Mathematics, the Pennsylvania State University.}

\begin{abstract}
Diagonalizing a Hamiltonian, which is essential for simulating its long-time dynamics, is a key primitive in quantum computing and has been proven to yield a quantum advantage for several specific families of Hamiltonians. Yet, despite its importance, only a handful of diagonalization algorithms exist, and correspondingly few families of fast-forwardable Hamiltonians have been identified. This paper introduces classical optimization algorithms for Hamiltonian diagonalization by formulating a cost function that penalizes off-diagonal terms and enforces unitarity via an orthogonality constraint, both expressed in the Pauli operator basis. We show that the landscape is benign: every stationary point is a global minimum, and any non-trivial stationary point yields a valid diagonalization, eliminating suboptimal solutions. Under a mild, numerically validated non-convexity assumption, we prove that the proposed optimization algorithm converges. In addition, we derive an a posteriori error bound that converts the optimization error directly into a bound on the Hamiltonian’s diagonalization accuracy.   We pinpoint a class of Hamiltonians that highlights severe drawbacks of existing methods, including exponential per-iteration cost, exponential circuit depth, or convergence to spurious optima. Our approach overcomes these shortcomings, achieving polynomial-time efficiency while provably avoiding suboptimal points. As a result, we broaden the known realm of fast-forwardable systems, showing that quantum-diagonalizable Hamiltonians extend to cases generated by exponentially large Lie algebras. On the practical side, we also present a randomized-coordinate variant that achieves a more efficient per-iteration cost than the deterministic counterpart. We demonstrate the effectiveness of  these algorithms through explicit examples and numerical experiments.
\end{abstract}

\maketitle

\section{Introduction}

With a provable computational advantage for short-time Hamiltonian simulation \cite{lloyd1996universal}, a variety of numerical schemes have been developed, including those based on the Trotter formula \cite{poulin2014trotter, tranter2019ordering, childs2019nearly, childs2021theory}, quantum signal processing (QSP) \cite{low2017optimal}, the quantum singular value transformation (QSVT) \cite{gilyen2019quantum}, and Taylor-series methods \cite{berry2015simulating}. While these schemes achieve optimal or near-optimal complexity for many Hamiltonians, the circuit complexity still grows at least linearly with the simulation time \(T\), as compatible with the no-fast-forwarding theorem \cite{atia2017fast}. Consequently, practical tasks, such as long-time simulations and high-precision energy measurements, remain out of reach.

This limitation has motivated efforts to identify Hamiltonians whose dynamics can be fast-forwarded or simulated with circuit complexity that is sublinear in \(T\) and polynomial in the number of qubits \(n\). Such treatment is often regarded as efficient diagonalization. Notable examples include Hamiltonians with quadratic fermionic operators and commuting local terms \cite{atia2017fast} as well as their extensions involving polynomially sized Lie algebras \cite{gu2021fast, patel2024extension, izmaylov2021define}, certain spin Hamiltonians \cite{kokcu2022fixed}, and Hamiltonians admitting efficient quantum-circuit diagonalizations \cite{novo2021quantum,atia2017fast}.

Together with the identification of efficiently-diagonalizable Hamiltonians, significant efforts have been dedicated to developing diagonalization algorithms. The primary objective of those algorithms is to output a quantum circuit that encodes a unitary corresponding to the resulting decomposition. Notable examples include eigen-decomposition \cite{gluza2024double,cirstoiu2020variational} and Cartan-decomposition \cite{kokcu2022fixed,somma2019unitary,Chu2024Lax,Wierichs2025Recursive}. Importantly, these algorithms can serve as a quantum compiler for more coherent Hamiltonian simulation than direct Trotterization. For instance, Somma \cite{somma2019unitary} introduced a Lie-algebraic diagonalization algorithm, reminiscent of Jacobi’s method, to prepare generalized coherent states.  K\"{o}kcu et al.\ \cite{kokcu2022fixed} proposed an optimization routine to compute a Cartan decomposition and demonstrated long-time Hamiltonian simulations for certain models. An improved, recursive Cartan decomposition method \cite{Wierichs2025Recursive} offers computationally-efficient unitary synthesis. In addition to these classical algorithms, a variational quantum algorithm by Cirstoiu et al.\ \cite{cirstoiu2020variational}, which optimizes a parameterized circuit to approximate an eigen-decomposition, showed Hamiltonian simulations longer than the Trotter method. As anticipated, long-time Hamiltonian simulations can facilitate the investigation of quantum many-body dynamics across a range of models, including time-dependent Hamiltonians, such as excitation dynamics with short-time external fields \cite{Wan2024VariationalCartan} and various spin chain dynamics models \cite{weinberg2017quspin}.

Despite these advances, significant computational challenges still persist in terms of compiling a unitary. For example, the Lie-diagonalization \cite{somma2019unitary, gu2021fast} requires a growing circuit depth to be inversely proportional to the target precision $\epsilon$ and scales linearly with the Lie algebra’s dimension. By contrast, the method in \cite{kokcu2022fixed} achieves fixed-depth circuits scaling at most polynomially for systems of polynomially sized Lie algebras, such as the transverse-field XY (TFTY) model. However, evaluating their cost function may incur exponential overhead due to exponentially many coefficient calculations: For the example of the TFXY model, it leads to \(\mathcal{O}\bigl(2^{n^2}\bigr)\) complexity.  The double-bracket approach in \cite{gluza2024double} in principle applies to more general Hamiltonians, but only a small number of iterations are allowed due to the fact that the circuit depth scales exponentially with the number of iterations. Variational methods \cite{cirstoiu2020variational},  on the other hand, rely on nonconvex optimization, and may converge to local minima or saddle points, and thus might be unreliable for long-time dynamics or large systems.

In this work, we introduce 
a new framework 
for Hamiltonian diagonalization. We first define a cost function, via Pauli strings,  to encode both off-diagonal terms and the orthogonality constraints. We then develop a gradient-descent method with an explicit normalization step to avoid trivial solutions. Notably, the per-iteration cost of our randomized variant scales quadratically with respect to the dimension of the Hilbert space, which is cheaper than the Lie-algebraic approaches \cite{kokcu2022fixed, somma2019unitary} in general. This framework, unlike the variational approaches, avoids suboptimal solutions,  as we prove that any nonzero stationary point of the cost landscape is a global minimum, corresponding to a valid diagonalization. Remarkably, if a proper Lie algebra associated to Hamiltonian is known in advance (as in \cite{kokcu2022fixed}), our formulation reduces to that setting, yet our parameterization remains more efficient. Under a mild nonconvexity condition, we establish sublinear or linear convergence guarantees and demonstrate that approximate cost minimization yields accurate diagonalizations of \(H\). We test both deterministic and randomized optimization algorithms in several numerical experiments to demonstrate the efficiency of our approach. Finally, we identify a family of Hamiltonians that pose substantial challenges for existing algorithms, e.g., incurring exponential per-iteration costs, requiring deep circuits, or converging to suboptimal solutions, while our algorithm maintains polynomial efficiency and converges to a global minimum, given a known unitary decomposition.

The rest of the paper is organized as follows. We first present an optimization formulation by introducing the cost function for Hamiltonian diagonalization. \cref{sec: main result} analyzes the properties of this cost function from an optimization perspective and studies both deterministic and randomized algorithms. In \cref{sec: example}, we show Hamiltonians that incur exponential per-iteration cost under existing Lie-algebraic methods \cite{kokcu2022fixed,somma2019unitary} but only polynomial cost with our algorithm. In \cref{sec: numerical test}, we present numerical results for several representative models.

\section{Method description}

Consider an \(n\)-qubit Hamiltonian of the form
\begin{equation}\label{eq: qubit H}
    H = \sum_{j=1}^{M}c_jQ_j,
\end{equation}
where $c_j\in\mathbb{R}$ and $Q_j$ denotes a Pauli string, i.e.,   a tensor product of $n$ Pauli matrices $Q_j=\otimes_{i=1}^n\sigma_j^{(i)}$  with $\sigma_j^{(i)}\in\{I,X,Y,Z\}$. Many Hamiltonians in chemistry and physics can be mapped to this form, and they satisfy the $M=\mathcal{O}(\text{poly}(n))$ count.

To diagonalize such Hamiltonians, thus enabling more efficient Hamiltonian simulations, we formulate an optimization problem. In what follows, we seek a unitary \(K\) and a diagonal matrix \(D\) such that
\[
    D = K^\dagger H K,
\]
with the obvious, but important implication that 
\(
    e^{-i t H} = K\,e^{-i t D}\,K^\dagger
\)
can be implemented by a circuit whose depth is independent of \(t\).

We seek such diagonalization by optimization. Before defining a cost function, we first recall that an expansion of a \emph{general} Hermitian matrix $A$ in terms of Pauli strings can be divided into diagonal and off-diagonal terms,
\begin{equation}
    A=\sum_{j=1}^{4^n}a_jP_j=\sum_{\forall i\in[n],\; \sigma_{j}^{(i)}\in\{I,Z\}}a_jP_j+ \sum_{\exists i\in[n],\;  \sigma_{j}^{(i)}\in\{X,Y\}}a_jP_j,
\end{equation}where $a_j= \frac{1}{2^n}\tr(AP_j)$. This expansion is derived from the fact that the set of Pauli strings forms an orthogonal basis of the set of matrices of size $2^n\times 2^n$ with respect to the Hilbert-Schmidt (HS) norm.  

Now, we consider a parameterized matrix $A=K^\dagger HK$ with $H$ being the given Hamiltonian and $K$ being parameterized. We define a cost function as follows
\begin{equation}\label{opt}
    f(\bm r,\bm \theta):=\sum_{P\in G_1}\mathrm{tr}\left(K(\bm r,\bm \theta)^\dagger HK(\bm r,\bm \theta)P\right)^2.
\end{equation}Here $G_1$, as a subset of the set of off-diagonal Pauli strings 
\[ G_1:=
\{\otimes_{i=1}^n\sigma^{(i)} :\sigma^{(i)}\in\{I,X,Y,Z\} \text{ and }\exists i\in[n],\;  \sigma^{(i)}\in\{X,Y\}\}
\]
 consisting of Pauli strings constituting the Pauli expansion of $K^\dagger HK$.  In \cref{opt},   we consider a parametric form of $K$, given by,
\begin{equation}  \label{K-param-form}
K(\bm r,\bm \theta) = \sum_{j=1}^{d} r_je^{i\theta_j} P_j, \quad  \bm r = (r_1,..r_d)\in\mathbb{R}^d,\quad \bm \theta = (\theta_1,..,\theta_d)\in\mathbb{R}^d.
\end{equation}
In our ansatz, the parameters \(\{r_j e^{i\theta_j}\}\) are exactly the Pauli‐basis coefficients of \(K\) written in polar decomposition form.  The integer \(d\) is the total number of Pauli strings in the ansatz  (and $2d$  is  the parameter dimension)  and it depends on the Hamiltonian under consideration.  
This definition implies that for any non-zero $\bm r$ such that $f(\bm r,\bm\theta)=0$, we obtain a diagonal matrix $K(\bm r,\bm\theta)^\dagger H K(\bm r,\bm\theta)$. 

Notice that $K$ in \cref{{K-param-form}} may not be necessarily unitary, unless we impose an orthogonality constraint on the cost function \eqref{opt}. A natural constraint is the Hilbert-Schmidt norm error between $K^\dagger K$ and the identity $I$, $\norm{K^\dagger K-I}_F^2$. Toward this end, we write the expansion of $K^\dagger K$ with \eqref{K-param-form} in Pauli strings,
\begin{equation}
    K^\dagger K = \sum_{P=P_iP_j,\;  i,j\in[d]}\phi_{P}(\bm r,\bm\theta) P,
\end{equation}where $\phi_P=\frac{1}{2^n}\tr(K^\dagger K P)$. Together, we define a total cost function with $2d$ parameters,
\begin{equation}\label{opt2}
    F(\bm r,\bm\theta):=f(\bm r,\bm\theta)+\sum_{P\in G_2}\phi_{P}(\bm r,\bm\theta)^2.
\end{equation}Here $G_2=\{P: P=P_iP_j,\;  i,j\in[d]\text{ and }P\neq I\}$. From \eqref{opt}, notice that the cost function \eqref{opt2} is represented by a sum of $\mathcal{O}(d^2M)$ terms in the worst-case scenario. We also note that the first part of this function accounts for the off-diagonal terms and the second part for the orthogonality condition. An immediate observation is that a non-zero global minimum ($\bm r\neq 0$) of the cost function in \eqref{opt2} always exists (e.g. an eigen-decomposition of $H$). In fact, we will show that any stationary point of the cost function \eqref{opt2} is a global minimum, and therefore any non-zero stationary point corresponds to a diagonalization unitary $K$ and an eigen-decomposition of $H$.

\section{Main result}\label{sec: main result}

In this section, we first present several results regarding the properties of the optimization problem \eqref{opt}. Then, we propose a non-convex optimization algorithm to achieve efficient diagonalization. In addition, we provide perturbation results that quantify the robustness of approximate solutions of our optimization problem \cref{opt2}. Our first main result states the main property of the optimization problem \eqref{opt2},
\begin{theorem}
   Any non-zero stationary point of the cost function \eqref{opt2} corresponds to an eigen-decomposition of the Hamiltonian.  
\end{theorem}

Next, we outline several properties of the optimization problem that lead to this benign property of the landscape.   

\subsection{Optimization landscape with all stationary points being global minima}

  The Lie-diagonalization algorithm in \cite{kokcu2022fixed} satisfies the property that any stationary point of their cost function yields a Cartan decomposition of $H$. We show that our cost function \eqref{opt2} has a similar property. More specifically, we prove that any non-zero stationary point of the optimization  \eqref{opt2} yields a unitary $K$ and a diagonal $K^\dagger HK$. We refer to a non-zero point as a point $(\bm r,\bm \theta)$ with $\bm r\neq 0$. 
  
  To proceed, we first work with the cost function \eqref{opt},  and show that for any of its non-zero stationary points, we obtain a diagonal matrix $K^\dagger HK$.  As a preliminary, we  derive an exact formula for the partial derivatives of the cost function \eqref{opt} as follows,
\begin{equation}\label{stograd}
\begin{split}
    \frac{\partial f}{\partial r_j}&=4\sum_{P\in G_1} \mathrm{tr}(K(\bm r,\bm\theta)^\dagger H K(\bm r,\bm\theta)P)\mathrm{Re}(\mathrm{tr}(HK(\bm r,\bm\theta) Pe^{-i\theta_j}P_j)),\\
    \frac{\partial f}{\partial \theta_j}&=-4\sum_{P\in G_1} \mathrm{tr}(K(\bm r,\bm\theta)^\dagger H K(\bm r,\bm\theta)P)\mathrm{Im}(\mathrm{tr}(HK(\bm r,\bm\theta)Pr_je^{-i\theta_j}P_j)).
\end{split}
\end{equation}
\begin{theorem}\label{thm: f}
     Any non-zero  stationary point of the cost function $f$ in \eqref{opt}, $(\bm r_c,\bm \theta_c)$ with $\bm r_c\neq 0$, is  a global minimum.
\end{theorem}
\begin{proof}
    Let $(\bm r_c,\bm\theta_c)$ be a non-zero stationary point, that is $\bm r_c\neq 0$. Denote $K_c=K(\bm r_c,\bm \theta_c)$. By \eqref{stograd}, it follows that at the stationary point,
    \begin{equation}\label{eq: real}     0=\frac{\partial f}{\partial r_j}=4\sum_{P\in G_1}\mathrm{tr}\left(K_c^\dagger HK_c P\right)\mathrm{Re}\left(\mathrm{tr}(HK_c Pe^{-i\theta_j} P_j)\right)=:4\mathrm{Re}\left(\mathrm{tr}\left(HK_cBe^{-i\theta_j}P_j\right)\right),
    \end{equation}
    where we have defined the matrix $B$
    \begin{equation}\label{eq: def of A}
        B:=\sum_{P\in G_1}\mathrm{tr}\left(K_c^\dagger HK_c P\right)P.
    \end{equation} Note that $B$ is Hermitian. By multiplying \eqref{eq: real} by $r_j$, summing over $j$'s, and using the definition of $K$ in \eqref{opt}, we arrive at
    \begin{equation}
        0 = \mathrm{Re}(\tr(K_c^\dagger HK_cB)) = \mathrm{Re}\left(\sum_{P\in 
 G_1}\mathrm{tr}(K_c^\dagger HK_c P)^2\right)=\sum_{P\in G_1}\mathrm{tr}(K_c^\dagger HK_c P)^2=f(\bm r_c,\bm \theta_c),
    \end{equation}by definition \eqref{eq: def of A}. Therefore, this completes the proof.

\end{proof}
According to \cref{thm: f}, any non-zero stationary point $K_c$ of the cost function \eqref{opt} provides us with a diagonal $K_c^\dagger HK_c$. However, this does not mean that $K_c$ is unitary. To eliminate such unphysical stationary points, we impose a constraint on the optimization problem \cref{opt}. As shown earlier in \eqref{opt2}, we introduce an orthogonality constraint as follows, 
\begin{equation}\label{eq: orthoconstraint}
    K^\dagger K = I.
\end{equation} This constraint can be reformulated in terms of parameters $\bm r$ and $\bm \theta$ as follows,
\begin{lemma}\label{lem: constraint}
    
The condition \cref{eq: orthoconstraint} is equivalent to the following constraint,
\begin{equation}
    \phi_P(\bm r,\bm\theta):=\sum_{j=1}^d  c_{j,P}r_jr_{j_P}\exp(i(\theta_j-\theta_{j_P}))=
    \begin{cases}
        1,\; j=j_P,\; c_{j,P}=1\; (P=I) \\ 
        0,\; j\neq j_P\; (P\neq I)
    \end{cases},
\end{equation}where $j_P$ is the index corresponding to $j$ such that $P_{j_P}P_{j}= c_{j,P}P$ with $c_{j,P}\in\{\pm 1,\pm i\}$.
\end{lemma}
\begin{proof}
    By the definition of $K$ in \eqref{opt}, we  consider the Pauli basis expansion of $K^\dagger K$ and check that
    \begin{equation}
        I=K^\dagger K = \left(\sum_{j=1}^d r_je^{-i\theta_j} P_j\right)\left(\sum_{j=1}^d r_je^{i\theta_j} P_j\right)=\sum_{P=P_iP_j,\;  i,j\in[d]}\phi_{P}(\bm r,\bm\theta) P,
    \end{equation}where $\phi_P$ is defined in the statement.  On the right hand side of this equation, the only remaining coefficient corresponds to the identity matrix, and the others are zeros, which recovers the constraint in the statement. 
\end{proof}
With this constraint, we improve \cref{thm: f} to guarantee that the resulting optimal $K_c$ is unitary and $K_c^\dagger HK_c$ is diagonal.
\begin{theorem}\label{thm: F}
     Any non-zero  stationary point of the total cost function in \eqref{opt2}, $(\bm r_c,\bm \theta_c)$ with $\bm r_c\neq 0$, is a global minimum. In addition, the matrix $K(\bm r_c,\bm\theta_c)$ is unitary up to a scaling factor, and consequently, 
     \begin{equation}
         H = \frac{1}{\norm{\bm r}^4}K(\bm r_c,\bm\theta_c)h(\bm r_c,\bm\theta_c)K(\bm r_c,\bm\theta_c)^\dagger,
     \end{equation}thereby yielding a KHK decomposition of $H$.
\end{theorem}
The proof of this theorem is shown in \cref{sec: Appendix A}. From the proof, we immediately see a property of the cost function \eqref{opt2}.
\begin{proposition}\label{prop: grad-cost}
The cost function in \eqref{opt2} satisfies that
    \begin{equation}
        \sum_jr_j\frac{\partial F}{\partial r_j} = 4F(\bm r,\bm \theta).
    \end{equation}
\end{proposition}
\begin{proof}
From \eqref{eq: real1} and \eqref{eq: F=0}, we see that the variables $\bm \theta$ do not affect the scaling of $\bm r$ and so the proof is straightforward.
\end{proof}

From \cref{thm: F},  minimizing the cost function in \cref{opt2} produces a unitary $K$ up to scaling and a KHK decomposition of the Hamiltonian $H$. However, we should avoid the trivial solution (e.g. $\bm r=0$). Namely, the desired optimization task is formulated as,
\begin{equation}
    \min_{\bm r\neq 0}F(\bm r,\bm\theta).
\end{equation} A simple observation is the self-similar property of the total cost function:  $F(s\bm r,\bm\theta)=s^4F(\bm r,\bm\theta)$ for any $s\in\mathbb{R}$. By this property and the fact that $\min_{\bm r\neq 0} F(\bm r,\bm\theta)=0$, the optimization problem can be equivalently reformulated as a constrained optimization on the unit sphere, while removing zero stationary points,
\begin{equation}\label{opt3}
    \min_{\bm r\neq 0}F(\bm r,\bm\theta) = \min_{\norm{\bm r}=1}F(\bm r,\bm\theta).
\end{equation}

\subsection{Optimization Algorithm based on the gradient descent}

In this section, we discuss optimization algorithms for solving the problem \eqref{opt3}, and present its convergence guarantee under a mild non-convex condition. Specifically, the algorithm updates the parameters $(\bm r,\bm \theta)$ at each iteration: At each iteration, the algorithm applies gradient descent (GD), followed by a normalization of the parameter $\bm r$ (in order to avoid the trivial solution $\bm r=0$). The simple algorithm is outlined in \cref{alg: algorithm1}. For convenience, we introduce the following notation in \cref{alg: algorithm1},
\begin{equation}
    \bm x_t:=(\bm r_t,\bm \theta_t), \quad \bm r(\bm x_t) = \bm r_t, \quad  \bm\theta (\bm x_t) = \bm \theta_t.
\end{equation}
\begin{algorithm}
\SetAlgoLined
	\KwData{initial guess $\bm x_0=(\bm  r_0, \bm\theta_0)$ with $\norm{\bm r_0}=1$,  learing rate schedule  $a_t\in(0,1)$, the maximal iteration number $T$, period $p$}
	\KwResult{approximate unitary $K(\bm r_T,\bm \theta_T)$ for a KHK decomposition of $H$}

	\For{$t=0:T$}{

             Construct $\nabla F(\bm x_t)$ from \eqref{opt2}\; 
            
            $\bs y_{t} = \bs x_{t}-a_{t}\nabla F(\bm x_t)$\;
            
            $\bm x_{t+1}=\left(\frac{\bm r(\bm y_t)}{\norm{\bm r(\bm y_t)}},\bm \theta(\bm y_t)\right)$\;
              }
	\caption{Deterministic optimization algorithm}
    \label{alg: algorithm1}
\end{algorithm}

On the other hand, unlike the conventional GD method,  \cref{alg: algorithm1} involves the step of normalization, and its convergence analysis should take into account the effect of normalization. We found through our analysis a similar convergence result to that of non-convex optimization using GD under a similar non-convex condition. Specifically, we use a weaker version of the Polyak--\L{}ojasiewicz (PL) condition that has been widely used in analyzing machine learning algorithms. A related, but more general, condition is the KL condition \cite[Definition 2.3]{li2018calculus}, which is a local condition and holds for any continuously-differentiable lower-bounded function. We slightly modify this condition such that the condition holds uniformly near global minima, which is stronger than the original condition. Nevertheless, we have validated the proposed condition with extensive numerical results. With this experience, we propose a non-convex condition as follows,
\begin{assumption}[{\bf Uniform KL condition}]\label{eq: kl condition}
    A function $F(\bm x)$ satisfies a uniform KL condition if there exist constants $\delta_f>0$, $\mu>0$ and $\alpha\in[0,2)$ such that  
    \begin{equation}
        \norm{\nabla F(\bm x)}^2\geq 4\mu F(\bm x)^{\alpha},
    \end{equation}whenever $F(\bm x)<\delta_f$. 
\end{assumption}
Under this assumption, we show that \cref{alg: algorithm1} converges to a non-trivial solution of the cost function in \cref{opt3} sublinearly or linearly with respect to $\epsilon$, depending on the exponent $\alpha$ determined by the given Hamiltonian and the initial guess.
\begin{theorem}\label{eq: convergence thm}
    Under \cref{eq: kl condition}, for any precision $\epsilon>0$ and a good initialization $\bm x_0$ such that $F(\bm x_0)<\delta_f$, \cref{alg: algorithm1} converges to a solution of the optimization problem \eqref{opt3} to precision $\epsilon$ sublinearly (if $\alpha\in(1,2)$) or linearly (if $\alpha\in[0,1]$) with the number of iterations at most,
    \begin{equation}
    T=\mathcal{O}\left(\frac{1}{ \epsilon^{\max\{\alpha-1,0\}}}\log \frac{1}{\epsilon}\right).
\end{equation}
\end{theorem}The proof of this theorem is shown in \cref{sec: Appendix A}.

\subsection{Sensitivity analysis for the approximate solution}

 Once a numerical solution is obtained from \cref{alg: algorithm1} with high accuracy, we show that a small optimization error implies a good diagonalization of $H$.  
 We provide such a posterior error bound as in the following lemma,
\begin{lemma}[A posterior error bound]\label{lem: perturb}
    For $(\bm r,\bm\theta)$ obtained from \cref{alg: algorithm1}, there exist a diagonal matrix $h_0$ and a off-diagonal matrix $\Delta\in\mathbb{C}^{2^n\times 2^n}$ such that 
    \begin{equation}\label{h0Delta}
        K(\bm r,\bm \theta)^\dagger H K(\bm r,\bm \theta) = h_0 + \Delta,
    \end{equation}with $\norm{\Delta}_F<\sqrt{\frac{F(\bm r,\bm \theta)}{2^n}}$. Furthermore, if $\norm{\bm r}=1$ {and the orthogonality constraint of $F(\bm r, \bm \theta)$ in \eqref{opt2} is less than $\frac{\epsilon}{2^n}$ with $\epsilon\leq \frac{1}{4}$, then, the following bound holds,
    \begin{equation}\label{eq: H-H}
        \norm{H-\widetilde{H}}_2\leq \frac{F(\bm r,\bm\theta)}{2^{n-1}}+6(1+\sqrt{F(\bm r,\bm \theta)})\norm{H}_F\sqrt{\epsilon},\quad \widetilde{H}:=K(\bm r, \bm\theta)h_0K(\bm r, \bm\theta)^\dagger.
    \end{equation}}
\end{lemma}The proof of this lemma is shown in \cref{sec: Appendix B}.
Roughly, this lemma guarantees a good approximate eigen-decomposition of $H$ for an optimized numerical solution of \eqref{opt2}. This result can further support the following perturbation relation between the eigenspaces of $H$ and $\widetilde{H}$,
\begin{equation}
    \norm{P_{\lambda_k(H)}-P_{k,\widetilde{H}}}_F\ll 1.
\end{equation}Here $P_{\lambda_k(H)}$ denotes the projector onto the $k$-th eigenspace of $H$ and $P_{k,\widetilde{H}}$ is defined to be the projector of the eigenspace corresponding to the eigenvalues of $\widetilde{H}$ closest to $\lambda_k$, satisfying that $\text{rank}(P_{k,\widetilde{H}}) =\text{rank}(P_{\lambda_k(H)})$. For each $k$, these projectors can be defined easily by listing the eigenvalues of $H$ and $\widetilde{H}$ in increasing order.  This can be useful for Hamiltonian simulations in low-energy subspace \cite{csahinouglu2021hamiltonian}. The following theorem shows that the perturbation relations hold uniformly for all $k$.    
\begin{theorem}\label{thm: perturb}
    For any $\epsilon'<\frac{\min_{\lambda_i\neq \lambda_j}\abs{\lambda_i-\lambda_j}}{4\max_k\text{rank}(P_{\lambda_k(H)})}$, we assume  a numerical solution $(\bm r,\bm \theta)$ of the cost function \eqref{opt2} as stated in \cref{lem: perturb} satisfying that
    \begin{equation}
            \frac{F(\bm r,\bm\theta)}{2^{n-1}}+6(1+\sqrt{F(\bm r,\bm \theta)})\norm{H}_F\sqrt{\epsilon}\leq \epsilon'.    
    \end{equation}Then the $k$-th eigenspace of $H$ is approximated by the set of eigenvectors of $\widetilde{H}$, whose corresponding eigenvalues are closest to $\lambda_k$, namely,  
    \begin{equation}
     \norm{P_{k,\widetilde{H}}-P_{\lambda_k(H)}}_F\leq C\left(\frac{F(\bm r,\bm\theta)}{2^{n-1}}+6(1+\sqrt{F(\bm r,\bm \theta)})\norm{H}_F\sqrt{\epsilon}\right),
 \end{equation}where the constant $C>0$ depends only on the spectrum of $H$.

\end{theorem}
The proof is shown in \cref{sec: Appendix B}.

\subsection{A randomized optimization algorithm}
The cost of evaluating the gradient of the cost function \eqref{opt2} in \cref{alg: algorithm1} scales as $\mathcal{O}(d^4M^2)$. To be specific, the cardinality of set $G_1$ scales as $\mathcal{O}(d^2M)$, since $K^\dagger HK$ in \eqref{opt} consists of $\mathcal{O}(d^2M)$ Paulis. Then, to evaluate a partial derivative in \eqref{stograd}, we execute a for-loop with $\mathcal{O}(d^2M)$ operations. Each operation computes the summand in \eqref{stograd} for a given $P\in G_1$, and the quantity $\tr(K^\dagger HKP)$ equals the coefficient of $K^\dagger HK$ for the Pauli string $P$, which requires $\mathcal{O}(dM)$ cost by looking for a pair of Pauli strings $(P_1,P_2)$ such that for each $Q_j$ in \eqref{eq: qubit H}, $P_1Q_jP_2$ equals $P$ up to a scaling factor. Since computing a partial derivative requires $\mathcal{O}(d^3M^2)$ and there are $\mathcal{O}(d)$ partial derivatives, in total, we obtain the gradient of the term \eqref{opt} in \eqref{opt2} with $\mathcal{O}(d^4M^2)$ operations. Similarly, we can deduce that $\mathcal{O}(d^3)$ operations are required to compute the gradient of the term for orthogonality constraint in \eqref{opt2} by using the formulations in \cref{lem: constraint} and \eqref{eq: real1}. To summarize, the per-iteration cost of \cref{alg: algorithm1} scales as $\mathcal{O}(d^4M^2)$. To circumvent this quartic scaling of per-iteration cost, we propose a random coordinate version of \cref{alg: algorithm1}. This random coordinate optimization quadratically reduces the per-iteration cost from $\mathcal{O}(d^4M^2)$ to  $\mathcal{O}(d^2M)$, which is a significant improvement. To this end, we derive recursive relations for evaluating the cost function and its gradient. The algorithm is quite similar to \cref{alg: algorithm1} except that a stochastic gradient estimate is introduced.
\begin{algorithm}
\SetAlgoLined
	\KwData{initial guess $\bm x_0=(\bm  r_0, \bm\theta_0)$ with $\norm{\bm r_0}=1$,  learing rate schedule  $a_t\in(0,1)$, the maximal iteration number $T$, period $p$}
	\KwResult{approximate unitary $K(\bm r_T,\bm \theta_T)$ for a KHK decomposition of $H$}

	\For{$t=0:T$}{

             Sample indices $\{i_t\}\subset [2^n]$ \; 

            Construct $\bm g_t=\sum \frac{\partial F(\bm x_t)}{\partial i_t} \bm e_{i_t} $
            
            $\bs y_{t} = \bs x_{t}-a_{t}\bm g_t$\;
            
            $\bm x_{t+1}=\left(\frac{\bm r(\bm y_t)}{\norm{\bm r(\bm y_t)}},\bm \theta(\bm y_t)\right)$\;
              }
	\caption{Random coordinate optimization algorithm}
    \label{alg: algorithm2}
\end{algorithm}Here we describe how to evaluate $\bm g_t$ and the cost function in this algorithm.

Intuitively, if in each iteration, only a small fraction of the components in the parameters $(\bm r,\bm \theta)$ are updated, one expects that the change in the cost function \eqref{opt2} and its partial derivatives would be small. Roughly, in the first part of the cost function \eqref{opt2}, we see that for sparse vectors $\delta \bm r, \delta \bm\theta$, 
\begin{equation}\label{eq: 1st part rcd}
    \tr(K(\bm r+\delta\bm r,\bm \theta + \delta\bm\theta)^\dagger HK(\bm r+\delta\bm r,\bm \theta + \delta\bm\theta)P) = \tr(K(\bm r,\bm \theta)^\dagger HK(\bm r,\bm \theta)P) + \Delta, 
\end{equation}where
\begin{equation}\label{eq: Delta}
\begin{split}
    \Delta &= 2\tr(K(\delta\bm r, \bm\theta + \delta\bm\theta)^\dagger H K(\bm r + \delta\bm r, \bm\theta + \delta\bm\theta)P) + \tr(K(\delta\bm r, \bm\theta + \delta\bm\theta)^\dagger H K(\delta\bm r, \bm\theta + \delta\bm\theta)P)\\
    & +  2\tr(K(\bm r, \delta\bm\theta)^\dagger H K(\bm r , \bm\theta + \delta\bm\theta)P) + \tr(K(\bm r, \delta\bm\theta)^\dagger H K(\bm r, \delta\bm\theta)P).
\end{split}
\end{equation}Notice that the four terms in the $\Delta$ involve at least one of the sparse vectors as input. If we say that $\delta \bm r$ and $\delta \bm \theta$ are $S$-sparse vectors, then the cost of evaluating the four terms scales as $\mathcal{O}(S)$. The following specific formulations validate this argument.  
\begin{equation}
    K(\delta\bm r, \bm\theta + \delta\bm\theta)^\dagger H K(\bm r + \delta\bm r, \bm\theta + \delta\bm\theta) = \left(\sum_{s=1}^S \delta r_{j_s}e^{i(\theta_{j_s}+\delta \theta_{j_s})}P_{j_s}\right)^\dagger H\left(\sum_{j=1}^d (r_j+ \delta r_j)e^{i(\theta_j+ \delta \theta_j)}P_j\right) 
\end{equation}
\begin{equation}
    K(\delta \bm r, \bm\theta + \delta \bm\theta)^\dagger H K(\delta \bm r, \bm\theta + \delta \bm\theta) = \left(\sum_{s=1}^S \delta r_{j_s}e^{i(\theta_{j_s}+ \delta \theta_{j_s})}P_{j_s}\right)^\dagger H \left(\sum_{s=1}^S  \delta r_{j_s}e^{i(\theta_{j_s}+ \delta \theta_{j_s})}P_{j_s}\right)
\end{equation}
\begin{equation}
    K(\bm r, \delta \bm\theta)^\dagger H K(\bm r , \bm\theta + \delta \bm\theta) = \left(\sum_{s=1}^S r_{j_s}(e^{i(\theta_{j_s}+ \delta \theta_{j_s})}-e^{i\theta_{j_s}})P_{j_s}\right)^\dagger H\left(\sum_{s=1}^S r_{j_s}e^{i(\theta_{j_s}+ \delta \theta_{j_s})}P_{j_s}+\sum_{j\not\in\{j_s\}_{s=1}^S}r_je^{i\theta_j}P_j\right)
\end{equation}
\begin{equation}
   K(\bm r, \delta \bm\theta)^\dagger H K(\bm r, \delta \bm\theta) = \left(\sum_{s=1}^S r_{j_s}(e^{i(\theta_{j_s}+\delta  \theta_{j_s})}-e^{i\theta_{j_s}})P_{j_s}\right)^\dagger H\left(\sum_{s=1}^S r_{j_s}(e^{i(\theta_{j_s}+\delta \theta_{j_s})}-e^{i\theta_{j_s}})P_{j_s}\right)
\end{equation}

Notice that in these formulas, there is at least one sum of $S$ terms multiplied left or right to the Hamiltonian $H$. When these formulations are computed with the trace with the Pauli basis element $P$ as in \eqref{eq: 1st part rcd},  only terms involving $P$ are counted. That is, we iterate over the sum with $S$ terms to find such terms, which takes only $S$ steps with a for-loop, thereby yielding $\mathcal{O}(S)$ computational cost. Together with these formulas, the relation \eqref{eq: 1st part rcd} provides us with efficient recursive calculations for the first part of the cost function \eqref{opt2} and its gradient, as roughly summarized as
\begin{equation}\label{eq: delta difference}
    f(\bm r_{t+1}, \bm \theta_{t+1}) = f(\bm r_t, \bm \theta_{t}) + \widetilde{\Delta}_t.
\end{equation} 
Here $\widetilde{\Delta}_t$ denotes the difference between $f(\bm r_{t+1}, \bm \theta_{t+1})$ and $f(\bm r_{t}, \bm \theta_{t})$, and is represented by the errors of the form $\Delta$ in \eqref{eq: Delta} and the definition \eqref{opt}, thereby depending on iteration step $t$. Estimating $\widetilde{\Delta}_t$ requires $\mathcal{O}(d^2MS)$ computational cost, recalling that the first part of the cost function \eqref{opt2} consists of $\mathcal{O}(d^2M)$ terms. Similarly, one can derive a recursive relation for the second part of the cost function \eqref{opt2} that involves the function $\phi_P(\bm r,\bm\theta)$. The cost of evaluating the gradient of the cost function \eqref{opt2} is then $\mathcal{O}(d^2MS^2)$, since we construct a $S$-sparse stochastic gradient, each partial derivative is computed using \eqref{eq: 1st part rcd}, which takes $\mathcal{O}(S)$, and it is repeated $\mathcal{O}(d^2M)$ times as the summation in \eqref{stograd} consists of $\mathcal{O}(d^2M)$ terms. This cost dominates that of computing the gradient of the second part of the cost function \eqref{opt2}. For more details, we refer the reader to our code, where recursive formulations are explicitly written.

With this randomization, the convergence analysis becomes non-trivial. Our algorithm performs the normalization on the radial parameters in $\bm r_t$ at each iteration. If the full gradient is calculated as in \cref{alg: algorithm1},  a lower bound of the updated parameters, $\bm r_{t+1}$, can be derived analytically using \cref{prop: grad-cost} as in the proof of \cref{eq: convergence thm}. However, in the case of randomized optimization, we cannot use \cref{prop: grad-cost}, and thus deriving a lower bound of the norm of $\bm r_{t+1}$ and a convergence guarantee is still an open issue. Nevertheless, with good initializations, the random coordinate algorithm \cref{alg: algorithm2} performs well and outputs a good approximation of a diagonalization of the Hamiltonian, as numerically shown in \cref{sec: numerical test}.

\section{A new family of efficiently diagonalizable Hamiltonians  }\label{sec: example}

Typical Hamiltonians in quantum computing applications are represented by a linear combination of polynomially many Pauli strings. In the form  \eqref{eq: qubit H} with $M=\mathcal{O}(\text{poly}(n))$, several families of Hamiltonians are found to be fast-forwarding or simulable on a quantum computer with polynomial efforts, such as polynomially sized Lie algebras \cite{gu2021fast, patel2024extension,kokcu2022fixed}, mutually commuting Hamiltonians \cite{atia2017fast}, and quantum diagonalizable Hamiltonians \cite{novo2021quantum}.

In this section, we identify families of Hamiltonians, and prove that their associated Lie algebras are exponentially large, but eigen-decompositions are polynomially small. To proceed, we follow the notion of quantum diagonalizable Hamiltonians \cite{cirstoiu2020variational,atia2017fast}, 
\begin{definition}[{\bf Quantum Diagonalizable Hamiltonian}]\label{def: QD}
A Hamiltonian is said to be quantum diagonalizable if
    \begin{equation}
        H = UDU^\dagger,
    \end{equation}where a unitary $U$ and a diagonal $D$ can be represented by polynomially sized circuits, respectively.
\end{definition}Within this family of Hamiltonians, we define a subset as follows,
\begin{definition}\label{def: QD1}
We consider a subset of the quantum diagonalizable Hamiltonians as follows,
    \begin{equation}
        H = UDU^\dagger,
    \end{equation}where a unitary $U$ and a diagonal $D$ can be represented by linear combinations of polynomially many Pauli strings, respectively.
\end{definition}
Within the set of Hamiltonians in \cref{def: QD1}, we show the existence of Hamiltonians whose Lie algebras are exponentially large. Specifically, the dimension of Lie algebra $\mathfrak{g}(H)$ is $4^n-1$, but the unitary of eigenvectors of the Hamiltonian $H$, $U$, is represented by $\mathcal{O}(\text{poly}(n))$ Pauli strings. The key idea is that the construction of the unitary of eigenvectors, $U$, involves $2n-1$ anticommuting Pauli strings, which is the maximum possible number for $n$-qubit systems \cite{bonet2020nearly}.  
\begin{example}\label{ex: Hams}
Consider a family of Hamiltonians defined as,
\begin{equation}
\begin{split}
    &H = UDU^\dagger,\\ 
    &U = \left(\prod_{m=1}^{\lfloor\log n\rfloor}U_m\right)\exp(i\theta Z_2)\left(c_0X_1Y_2+c_1Z_1Y_2+c_2Z_2+\sum_{j=3}^nc_jX_2Y_3...Z_j\right),\\
    &D = I+\sum_{j=1}^nd_jY_j,
\end{split}
\end{equation}
where all parameters are non-zero, especially $\theta\neq k\pi$ for integer $k$,  $c_j$'s are real and $\sum_{j=0}^nc_j^2=1$. Here, each $U_m$ can be a Pauli rotation or the basic Clifford gate among $S$, $H$, and $CNOT$ and $U_{\lfloor\log n\rfloor}\neq \exp(-i\theta Z_2)$.     
\end{example}
For any Hamiltonian in \cref{ex: Hams}, we show that $\mathfrak{g}(H)=\text{su}(2^n)$, which yields that $\text{dim}(\mathfrak{g}(H))=4^n-1$. The proof is shown in \cref{sec: Appendix C}.

Using \cref{ex: Hams}, one can also conceive similar examples with logarithmic repetitions as follows, which results in Hamiltonians whose Lie algebras are exponentially large,  
\begin{example}\label{ex: Hams1}
Appending unitaries of the form in \cref{ex: Hams},
    \begin{equation}
        U = \prod_{r= 1}^{\lfloor \log n \rfloor}U_r,
    \end{equation}reproduces Hamiltonians whose Lie algebras are exponentially large but polynomially small in Pauli strings,
    \begin{equation}
        H = UDU^\dagger,
    \end{equation}where $D$ is defined in \cref{ex: Hams}.
\end{example}
 Here we exploited the property of anticommuting Pauli strings \cite{izmaylov2019unitary} to conceive Hamiltonians whose Lie algebras are isomorphic to $\text{su}(2^n)$ in the perspective of dynamical Lie algebra \cite{smith2024optimally}.  We could find another family of Hamiltonians by employing a complete pool of Pauli strings in the adaptive variational quantum eigensolver \cite{tang2021qubit}. Since the set of Pauli strings there are used for preparing arbitrary real ground states, exponentially sized Lie algebras of Hamiltonians would be obtained. Together, one may be able to conceive a variety of Hamiltonians, whose eigen-decomposition is polynomially represented in Pauli strings as \cref{def: QD1}, but the associated Lie algebra is exponentially large, following the idea of extensions of polynomially-sized Lie algebras in \cite{patel2024extension}. For such Hamiltonians, the per-iteration computational cost of \cref{alg: algorithm1} scales polynomially, since the matrices $K^\dagger HK$ and $K^\dagger K$ involved in \eqref{opt2} are represented by polynomially many Paulis. By contrast,  the Lie-diagonalization approaches \cite{gu2021fast,kokcu2022fixed} would require an exponential scaling of per-iteration cost in this case, due to Lie algebra being exponentially large.

\section{Numerical result}\label{sec: numerical test}

In this section, we provide several numerical results to test the proposed algorithms. Our optimization algorithms rely on elementary calculations with Pauli strings, such as commutator calculations, conversion from Pauli strings to quaternary representations and vice versa. For these calculations, we used and extended a part of the code  \cite{kokcu2022fixed} in our optimization algorithms. First, we demonstrate \cref{alg: algorithm1} for several Hamiltonians that are randomly generated in such a way that the unitary matrix $K$ containing the eigenvectors and the diagonal matrix from the eigenvalues are sparse in terms of Pauli basis expansion. 

 \cref{fig:random H} shows the optimization results for four Hamiltonians represented by different numbers of Pauli strings $\text{len}(H)=16, 24, 64, 96$, where we denote $\text{len}(H)$ to be the number of Pauli strings in the Hamiltonian $H$. {To test \cref{alg: algorithm1}, we construct the Hamiltonians with the following steps (which may not belong to the families in \cref{ex: Hams} or \cref{ex: Hams1}):  First, we sample a pair of polynomial numbers of Pauli strings, say, $(\{P_{i_a}\}_{i_a\in I_a},\{P_{i_b}\}_{i_b\in I_b})$, where $\abs{I_a}=\mathcal{O}(\text{poly}(n))$ and $\abs{I_b}=\mathcal{O}(\text{poly}(n))$.  Here, each Pauli string $P_{i_a}$ is diagonal. Second, we sample the same number of pairs of random real numbers $(\{c_{i_a}\}_{i_a\in I_a}, \{c_{i_b}\}_{i_b\in I_b})$. With these, a diagonal matrix is constructed as $D=\sum_{i_a\in I_a}c_{i_a}P_{i_a}$, and a unitary matrix is formulated as the product of Pauli rotations, say,  $U=\prod_{i_b\in I_b}\exp(i c_{i_b}P_{i_b})$. Lastly, with this pair of unitary and diagonal, we construct a Hamiltonian that is represented by $UDU^\dag$. For example, the Hamiltonian comprised of $16$ Pauli strings is represented by a pair of numbers of Pauli strings $(4,2)$ for its eigen-decomposition. Similarly, for the cases where $24,64,96$ Pauli strings are used to express Hamiltonians, the diagonals and the unitaries in their eigen-decompositions are represented by $(4,4)$, $(6,2)$, $(6,4)$ pairs of numbers of Pauli strings, respectively. That is, for the cases $\text{len}(H)=16,24,64,96$, the number of parameters in $K$ in \eqref{opt} is $8,32,8,32$, respectively.  }

 We observe that in all cases, the algorithm \cref{alg: algorithm1} converges rather quickly. To study the convergence in response to \cref{eq: convergence thm},  we monitored the value of $\alpha$ in \cref{eq: kl condition} during the iteration. {On average, the values of $\alpha$ for the cases where $\text{len}(H)=16,24,64$ remain close to $1$. In these cases, the optimization results are shown to converge almost linearly as in \cref{fig:random H}, which is consistent with the theoretical result in \cref{eq: convergence thm}. On the other hand, in the case where $\text{len}(H)=96$, the value of $\alpha$ seems to deviate from $1$ and the convergence becomes sublinear as illustrated in \cref{fig:random H}, again consistent with \cref{eq: convergence thm}. }

To further illustrate the optimization properties of our algorithm, we study four-qubit examples, including an XXZ model and a Hubbard model, and the results are shown in \cref{fig: XXZ HB}. In the numerical tests, four different initializations are set by changing the parameters in the models. Specifically, the XXZ model is given by
\begin{equation}\label{H-XXZ}
    H_{XXZ} = J\sum_{i=1}^{n-1}(X_iX_{i+1}+Y_iY_{i+1})+\Delta \sum_{i=1}^{n-1} Z_iZ_{i+1},
\end{equation}and
for the Hubbard model, the Hamiltonian is expressed as
\begin{equation}\label{H-Hubbard}
    H = -\frac{t}{2}\sum_{\sigma\in\{\uparrow,\downarrow\}}\sum_{j=1}^{n-1}\left(X_{j+1,\sigma}X_{j,\sigma}+Y_{j+1,\sigma}Y_{j,\sigma}\right)+\frac{U}{4}\sum_{j=1}^n\left(-Z_{j,\uparrow}-Z_{j,\downarrow}+Z_{j,\uparrow}Z_{j,\downarrow}\right).
\end{equation}Here we omit the identity matrix, which does not affect diagonalization.

For the XXZ model, we aim to diagonalize the Hamiltonian with parameters $J=1$. For the Hubbard model, we set  $t=1$. {Meanwhile, as shown in \cref{tab: filtering}, we choose several different values for  $\Delta$ or $U$ and run the optimization algorithms. Specifically, with a Hamiltonian defined by a different value of $\Delta$, for instance, we compute its eigen-decomposition using an eigensolver, obtain a unitary $K$, expand $K$ in terms of Pauli strings as \eqref{opt} and use the coefficients as the initial parameters for the optimization \eqref{opt2} with the Hamiltonian defined by the target value of $\Delta$.}

In these experiments, we used \cref{alg: algorithm2} since the Pauli expansion of the unitary of eigenvectors and the diagonal of eigenvalues are not sparse in these cases, and \cref{alg: algorithm1} may not be efficiently implementable. 
As shown in \cref{fig: XXZ HB},  the numerical results suggest the fast convergence of \cref{alg: algorithm2} to diagonalizations of the target Hamiltonians,  especially when the initial parameter is well selected. These findings suggest that the optimized parameter set obtained from \eqref{opt2} for one Hamiltonian is an excellent initialization for our optimization procedure when diagonalizing another Hamiltonian with slightly perturbed physical parameters.

\begin{figure}[htbp]
    
    %\centering

    \begin{center}
   \begin{subfigure}[b]{0.45\textwidth}
    \centering
    \includegraphics[width=\textwidth]{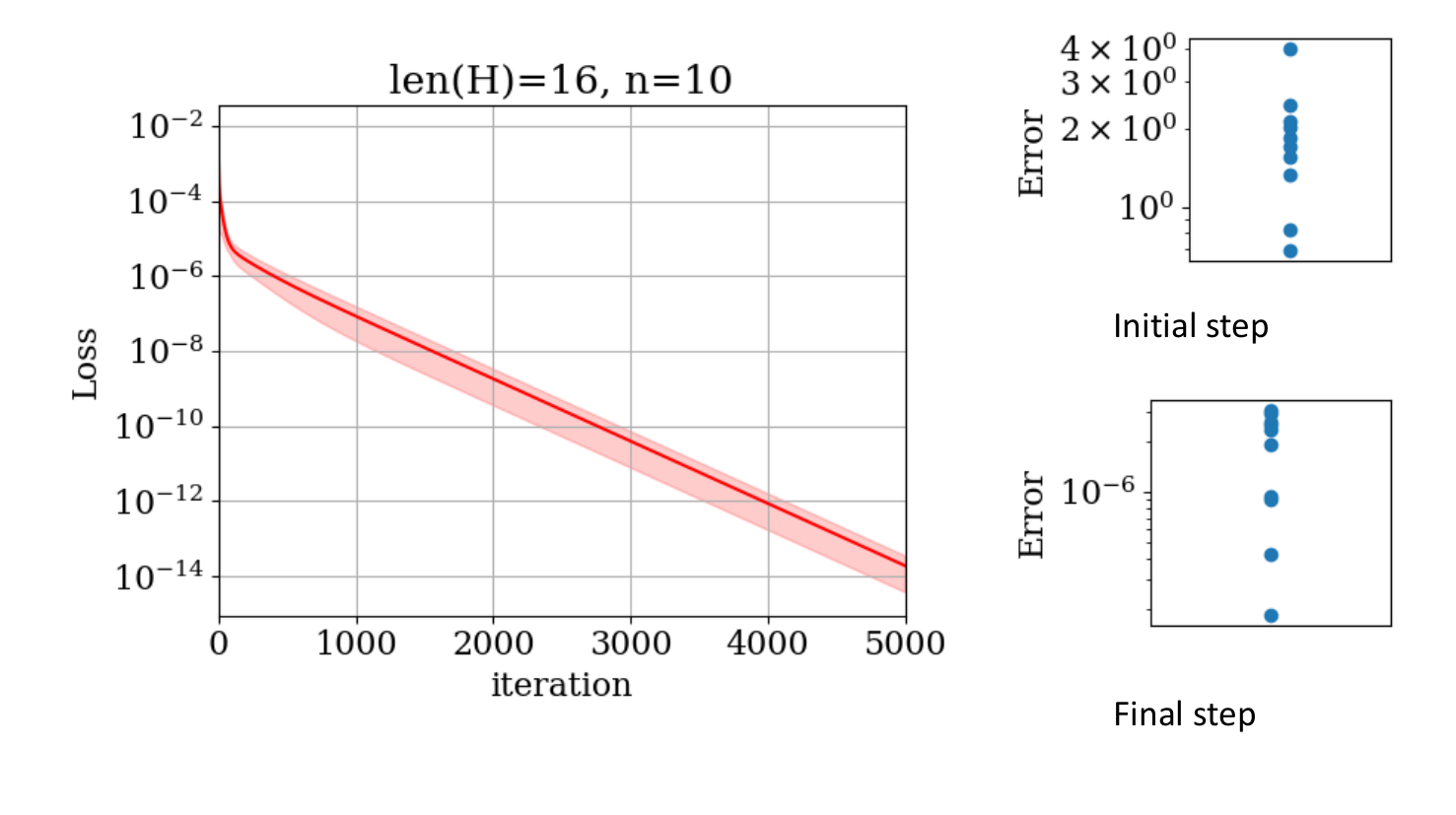}
    \end{subfigure}
    \begin{subfigure}[b]{0.45\textwidth}
    \centering
    \includegraphics[width=\textwidth]{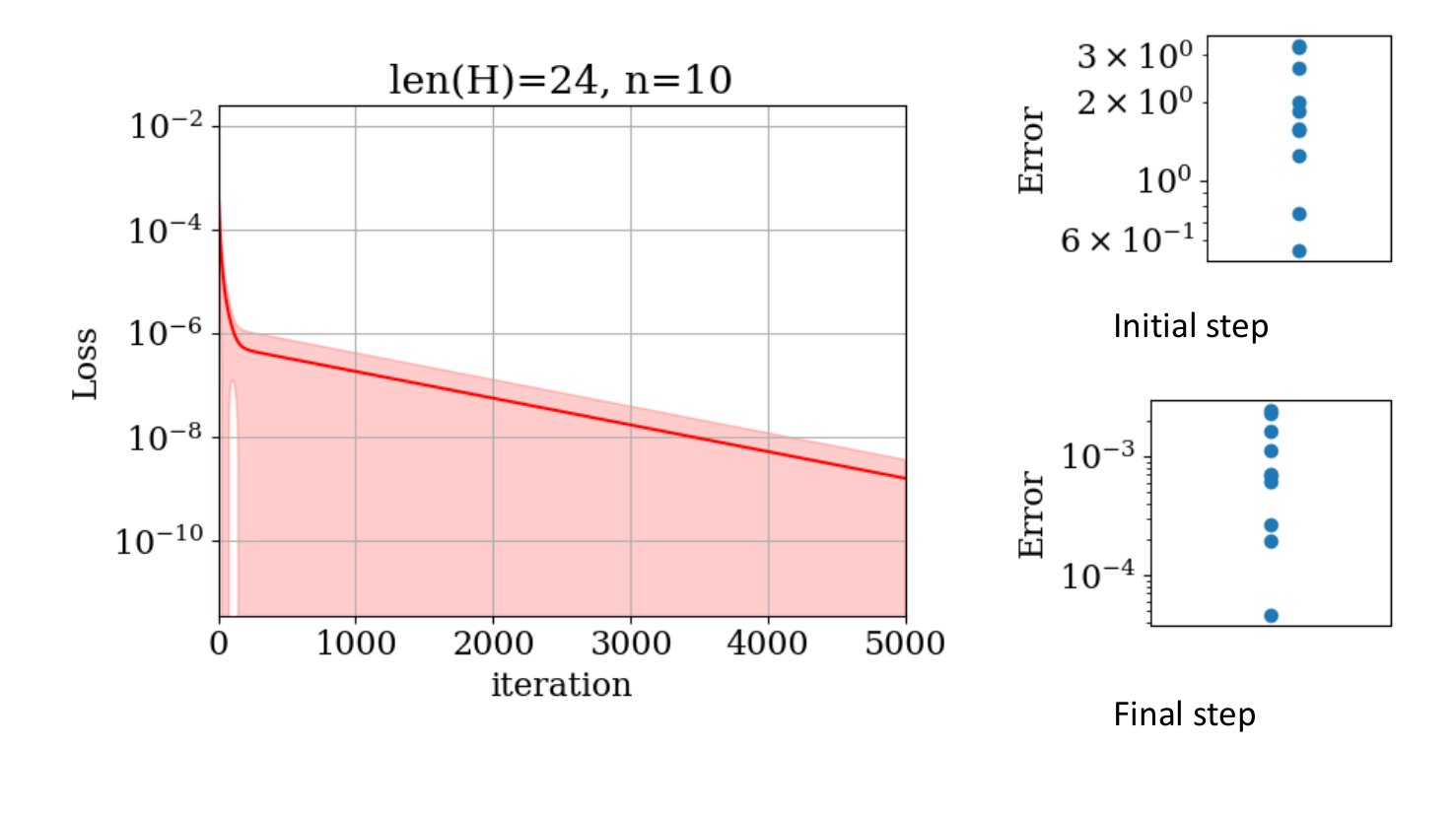}
    \end{subfigure}
    \begin{subfigure}[b]{0.45\textwidth}
    \centering
    \includegraphics[width=\textwidth]{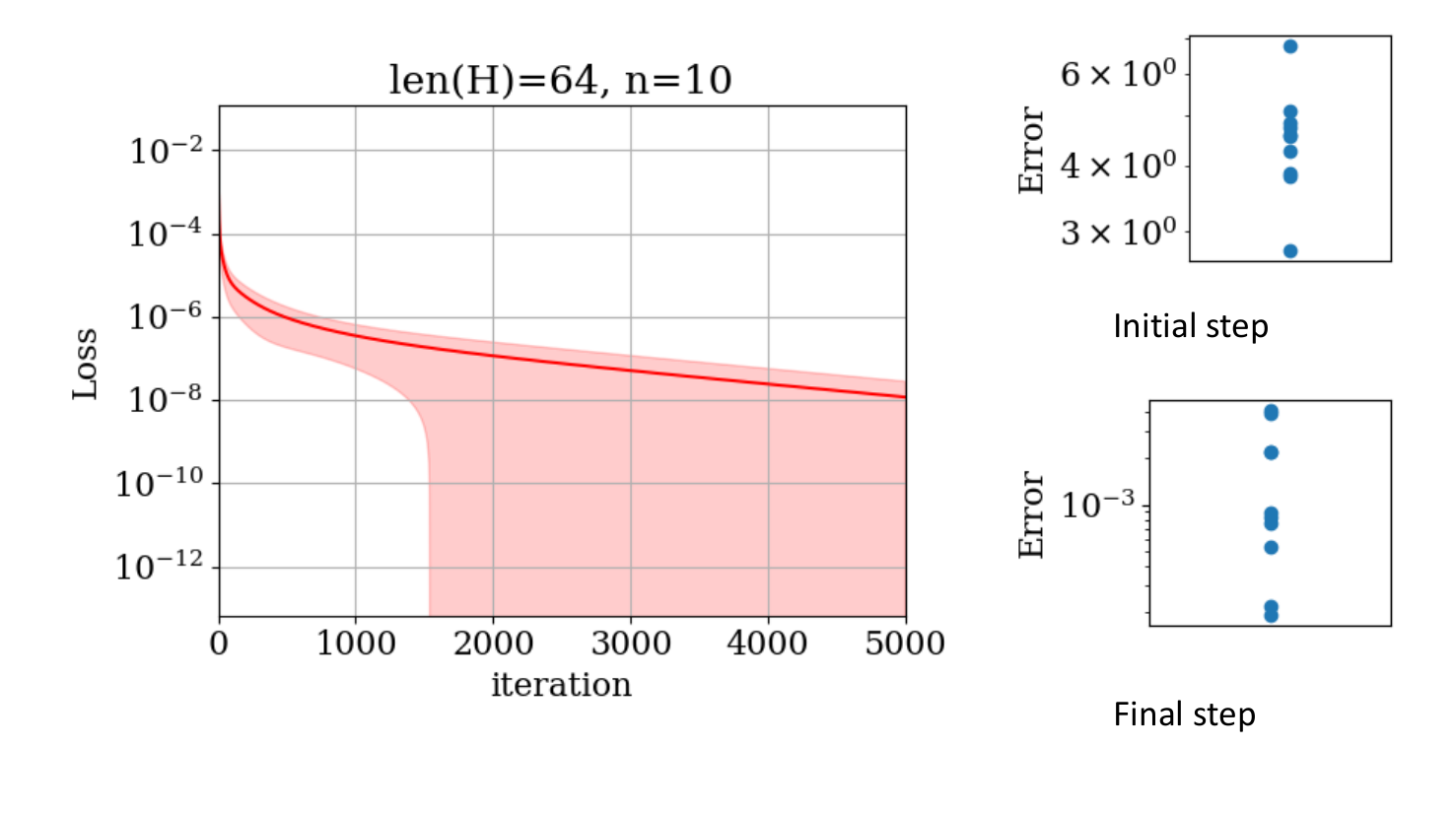}
    \end{subfigure}
    \begin{subfigure}[b]{0.45\textwidth}
    \centering
    \includegraphics[width=\textwidth]{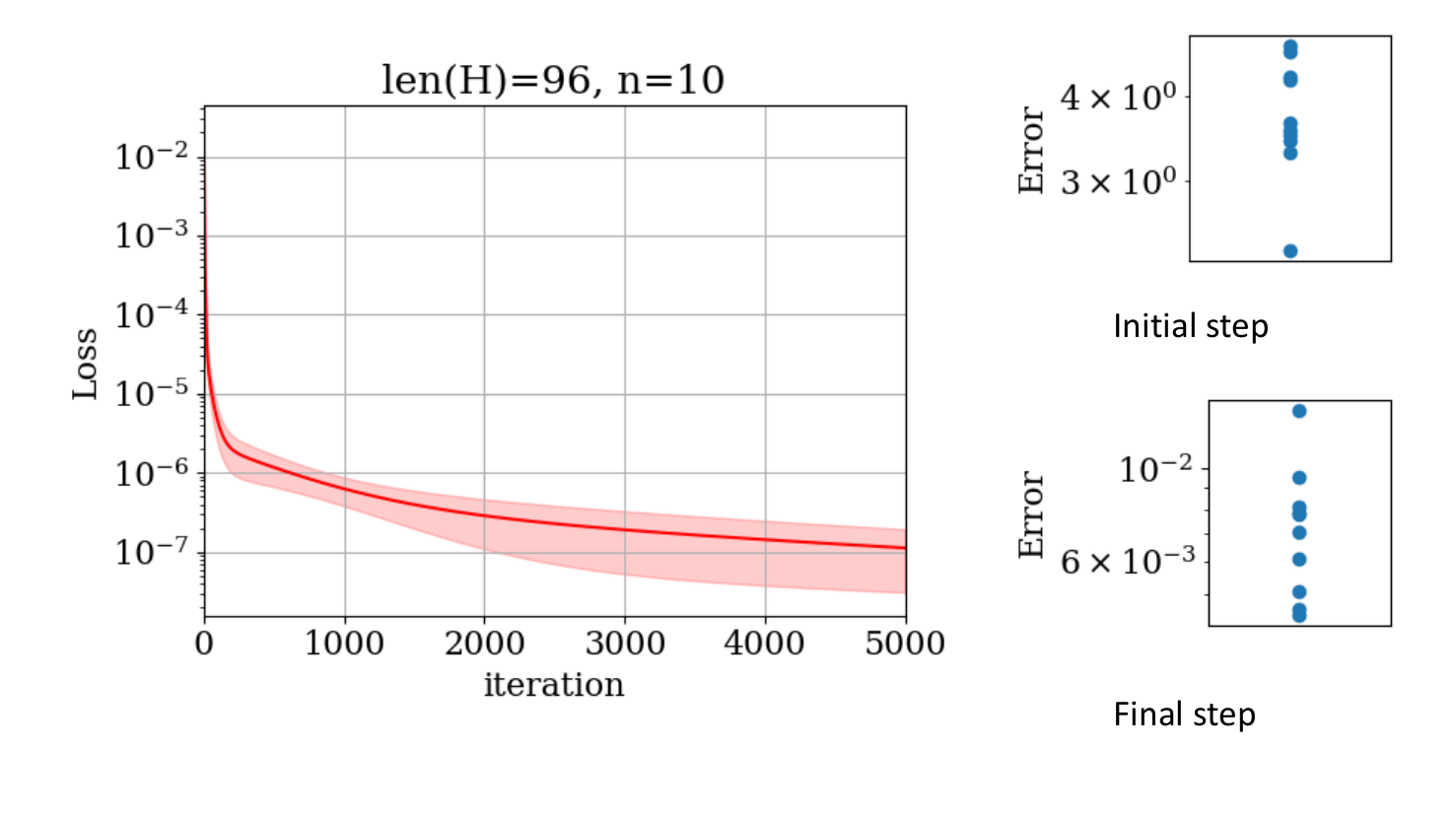}
    \end{subfigure}

    \begin{subfigure}[b]{0.3\textwidth}
    \centering
    \includegraphics[width=\textwidth]{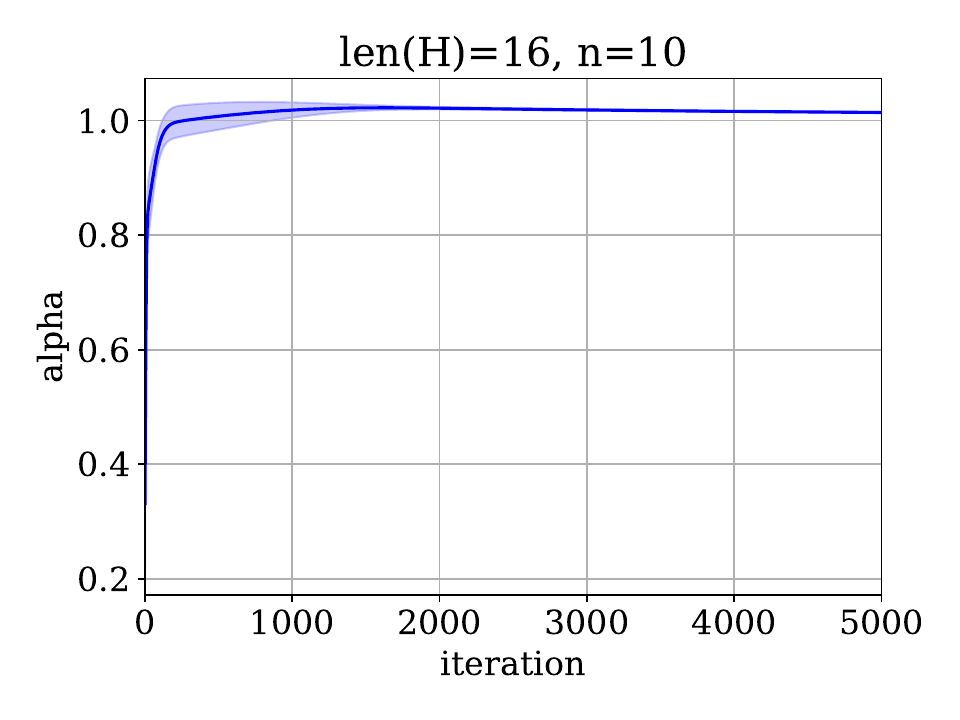}
    \end{subfigure}
    \begin{subfigure}[b]{0.3\textwidth}
    \centering
    \includegraphics[width=\textwidth]{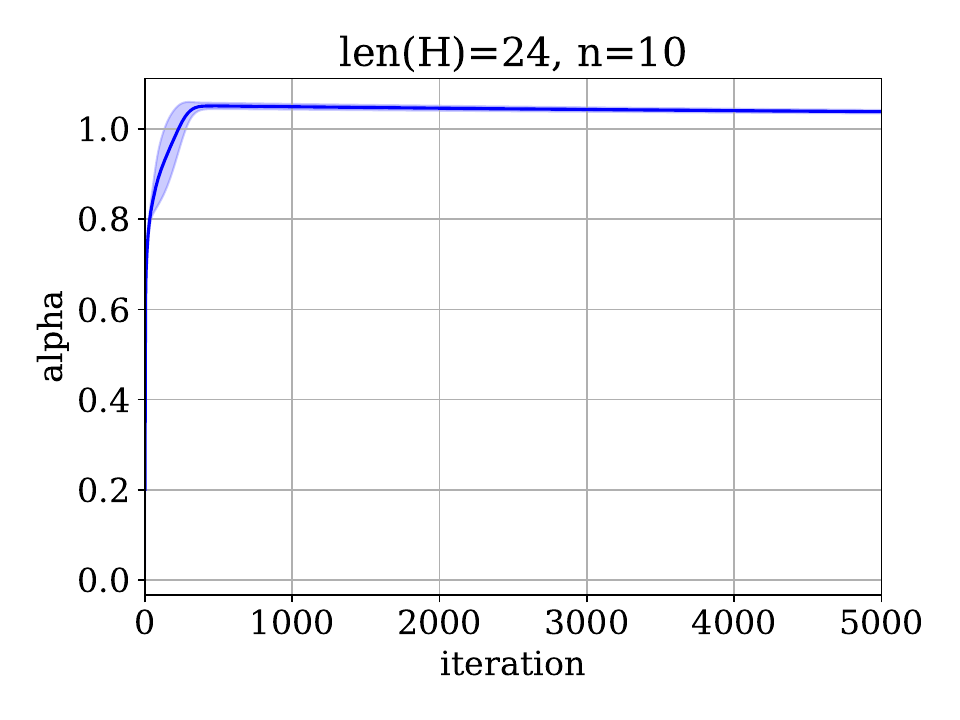}
    \end{subfigure}
    
    \begin{subfigure}[b]{0.3\textwidth}
    \centering
    \includegraphics[width=\textwidth]{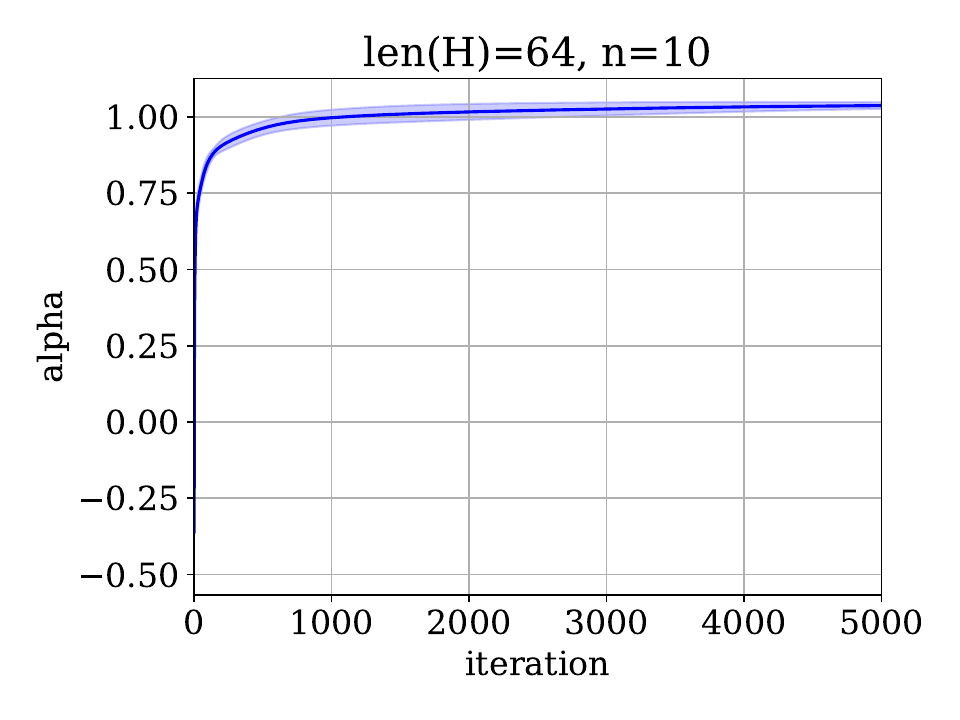}
    \end{subfigure}
    \begin{subfigure}[b]{0.3\textwidth}
    \centering
    \includegraphics[width=\textwidth]{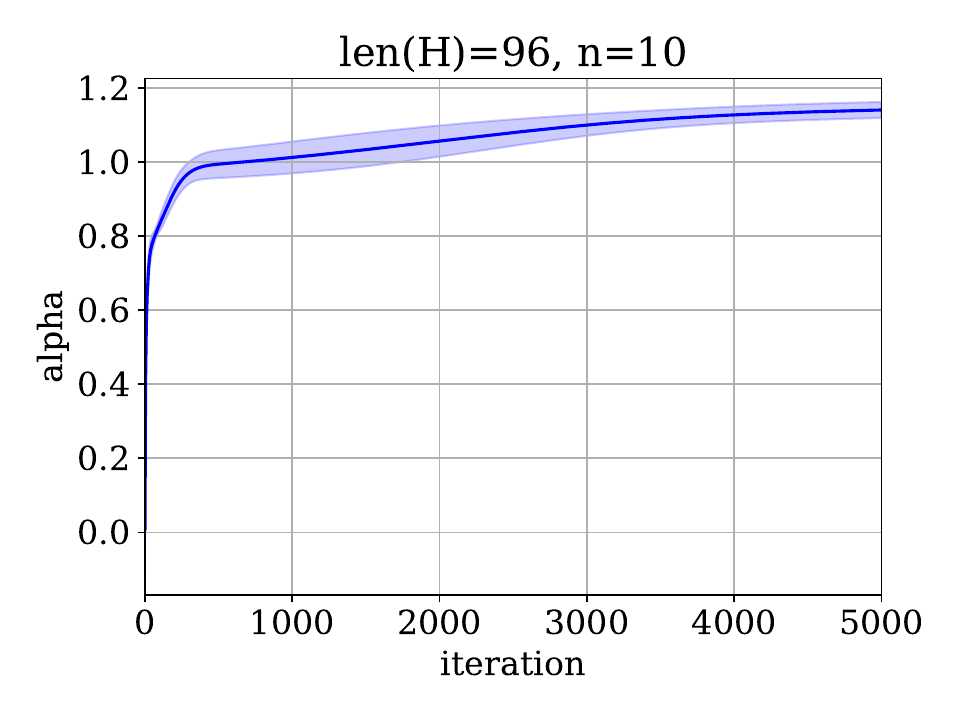}
    \end{subfigure}
\end{center}

    \caption{Optimization results (Top) and the corresponding values of $\alpha$ with $\mu =1$ in the KL condition \eqref{eq: kl condition} (Bottom) for four different Hamiltonians randomly generated comprised of $16$, $24$, $64$ and $96$ Pauli strings for $n$-qubit system ($n=10$). The error in the optimization results indicates the Frobenius error, $\norm{H-\widetilde{H}}_F$, where $H$ and $\widetilde{H}$ denote the original Hamiltonian and an approximation defined in \cref{lem: perturb}. We compare the errors estimated at the first step and the last step of the optimization. }
    \label{fig:random H}
\end{figure}

\begin{figure}[htbp]

    \centering
   \begin{subfigure}[b]{0.45\textwidth}
    \centering
    \includegraphics[width=\textwidth]{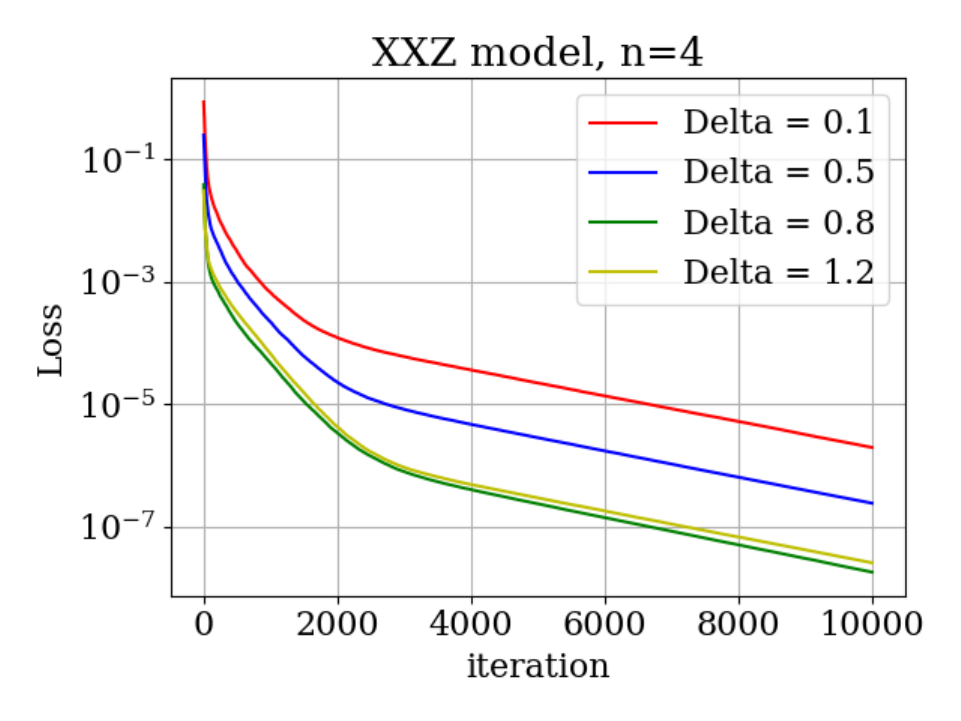}
    \end{subfigure}
    \begin{subfigure}[b]{0.45\textwidth}
    \centering
    \includegraphics[width=\textwidth]{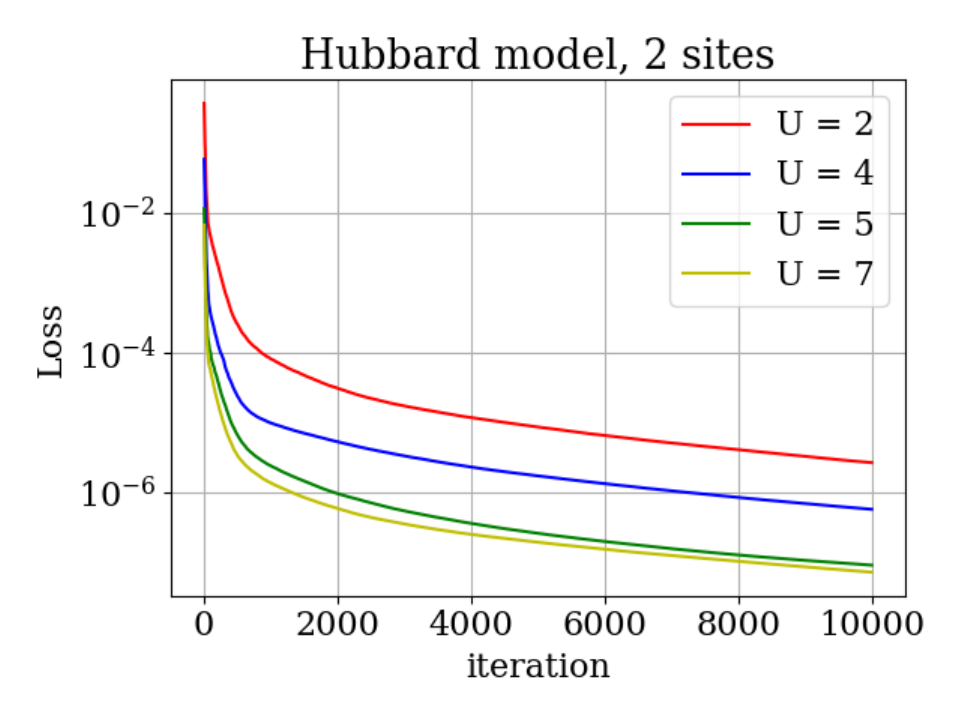}
    \end{subfigure}

    \caption{Optimization results for XXZ model and Hubbard model with four different choices of parameters $\Delta$ or $U$ in \cref{H-XXZ,H-Hubbard} as listed in \cref{tab: filtering}, which give four different initializations.  }
    \label{fig: XXZ HB}
\end{figure}

\begin{table}
    \centering
    \begin{subtable}[h]{0.3\textwidth}
        \centering
         \begin{tabular}{c|c|c}
         $\Delta$ & Initial error & Final error\\
        \hline
         0.1 & 3.870 & 0.028 \\
         \hline
         0.5 & 2.073 & 0.010 \\
         \hline
         0.8 & 0.800 & 0.003 \\
         \hline
         1.2 & 0.765 & 0.003 \\
    \end{tabular}
    \caption{XXZ model}
    \end{subtable}
    \begin{subtable}[h]{0.3\textwidth}
        \centering
         \begin{tabular}{c|c|c}
         $U$ & Initial error & Final error\\
        \hline
         2 & 2.529 & 0.042 \\
         \hline
         4 & 1.000 & 0.020 \\
         \hline
         5 & 0.441 & 0.008 \\
         \hline
         7 & 0.350 & 0.007 \\
    \end{tabular}
    \caption{Hubbard model}
    \end{subtable}
    \caption{Initial and final errors, $\norm{H-\widetilde{H}}_F$, for the optimization results in \cref{fig: XXZ HB}}
    \label{tab: filtering}
\end{table}

\section{Discussion}

In the previous sections, we identified a family of fast-forwarding Hamiltonians (\cref{def: QD1}) and developed an efficient classical algorithm (\cref{alg: algorithm1}). In this section, we examine the significance of our results from various perspectives and discuss related work.

The previous works \cite{cirstoiu2020variational,novo2021quantum} have raised an important open question: is the set of quantum diagonalizable Hamiltonians (\cref{def: QD}) smaller than that of fast-forwarding Hamiltonians? Addressing this question may first require identifying the Hamiltonians that lie within the family of quantum diagonalizable Hamiltonians (\cref{def: QD}).  To date, much of the existing literature has focused on Hamiltonians associated with polynomially sized Lie algebras (PLA), such as quadratic fermionic and certain bosonic Hamiltonians \cite{gu2021fast}, non-contextual Pauli and fermionic mean-field Hamiltonians \cite{patel2024extension}, and other further examples thanks to the existence of a polynomially sized ansatz with Pauli rotations \cite{kokcu2022fixed}. In contrast, the case of Hamiltonians generated by exponentially large Lie algebras has remained largely unexplored in this context, despite the fact that such Hamiltonians encompass a wide range of physically and computationally significant models \cite{wiersema2024classification}. This raises an important question: Is every quantum diagonalizable Hamiltonian necessarily associated with a PLA? From the examples presented in \cref{ex: Hams} and \cref{ex: Hams1}, we provide evidence that the class of quantum-diagonalizable Hamiltonians \emph{strictly} contains the PLA Hamiltonians. As illustrated in \cref{fig:diagram}, this result establishes a previously unrecognized relationship between these two classes. Namely, the former is strictly larger than the latter.  In other words, there remains an opportunity to construct efficient diagonalization circuits for certain fast-forwardable Hamiltonians whose underlying Lie algebras are exponentially large, beyond the scope of conventional PLA-based approaches.

Second, we introduced a classical optimization algorithm (\cref{alg: algorithm1}) tailored to the family of Hamiltonians in \cref{def: QD1}.  While, in principle, existing quantum algorithms such as the variational fast-forwarding (VFF) method \cite{cirstoiu2020variational} and the double-bracket diagonalization algorithm \cite{gluza2024double} are applicable to this family, several limitations hinder their practical use. First, the optimization landscape of VFF may involve local minima and saddle points as other variational quantum algorithms. Although strategies such as adaptive ansatz growth (e.g., adaptive VQE) have been proposed to mitigate these issues, they do not offer guarantees against convergence to suboptimal points in general, and may lead to over-parameterization in practice \cite{grimsley2023adaptive,feniou2023overlap}. Furthermore, implementing the coherent Hamiltonian simulations, which is required for the training, introduces additional overhead and complexity. By contrast, our algorithm optimizes the cost function that is analytically guaranteed to have only global minima as stationary points according to \cref{thm: F}. Specifically, any non-zero global minimum directly yields a valid diagonalization unitary, and our algorithm converges to only a non-zero minimum due to the normalization step enforced at each iteration. As a result, our approach circumvents the suboptimality issues that are inherent in VFF, offering a more direct and reliable route to diagonalization, without any quantum resource.  Another approach by Gluza \cite{gluza2024double} implements a double-bracket dynamics to construct a diagonalizing unitary. This technique is constructive and broadly applicable, due to its non-increasing property of the off-diagonal terms. However, the family of Hamiltonians for which the algorithm operates efficiently remains unclear from a complexity-theory perspective.  Moreover, the required circuit depth increases exponentially with the number of iterations, posing serious challenges to scalability in practical implementations.  Therefore, we think that the proposed classical diagonalization algorithm is more suitable than the two quantum algorithms for the family of Hamiltonians considered. This conclusion aligns with insights from the literature of dequantization algorithms, where certain classical algorithms, under additional assumptions, are shown to perform comparably to their quantum counterparts \cite{tang2019quantum, gharibian2022dequantizing}, and thus reshape expectations around quantum advantage in specific algorithmic contexts. 

\begin{figure}
    \centering
    \includegraphics[width=0.9\linewidth]{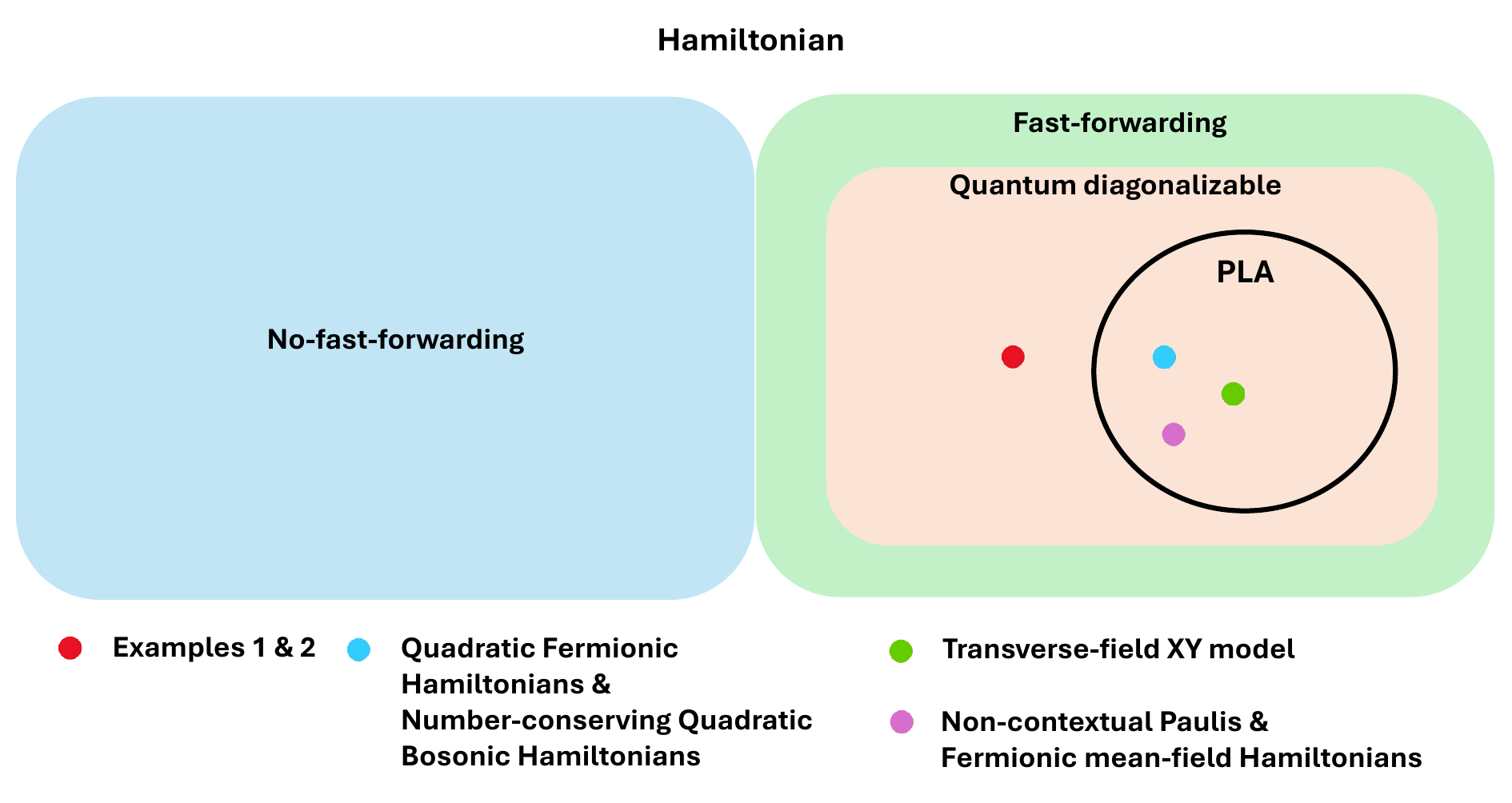}
    \caption{The diagram illustrates examples of quantum diagonalizable Hamiltonians: this work (red), \cite{gu2021fast} (blue), \cite{kokcu2022fixed} (green), and \cite{patel2024extension} (purple). }
    \label{fig:diagram}
\end{figure}

We remark that the formulation \eqref{opt2} is equivalent or similar to previous Lie-diagonalization approaches. For example, Somma's Lie-diagonalization algorithm \cite{somma2019unitary,gu2021fast} is viewed as a variant of Jacobi’s eigenvalue algorithm, iteratively minimizing off-diagonal terms, which is conceptually aligned with \cref{alg: algorithm1}. However, their algorithm is more suitable for PLA Hamiltonians, whereas our algorithm is designed to handle a family of Hamiltonians with exponentially large Lie algebras (e.g. \cref{def: QD1}). On the other hand, the algorithm in \cite{kokcu2022fixed} finds a Cartan decomposition of Hamiltonian, which has been widely studied in the context of quantum control theory \cite{d2007quantum,earp2005constructive,khaneja2001cartan}. Their algorithm assumes prior algebraic knowledge of the Cartan decomposition of a given Hamiltonian and essentially seeks to identify a parameterized unitary $K$ such that
\begin{equation}\label{eq: conditions}
    \begin{split}
        [K^\dagger HK,v]&=0,\quad \text{ for some }v\in\mathfrak{h},\\
        K^\dagger HK&\in\mathfrak{m},
    \end{split}
\end{equation}as deduced from the proof of \cite[Theorem.5, Appendix D]{kokcu2022fixed}. Here, a Pauli string $v$ satisfies that its corresponding exponential $e^{itv}$ for real $t$ is dense in the Lie group, $e^{i\mathfrak{h}}$, associated to a Cartan subalgebra $\mathfrak{h}$. However, the condition in \cref{eq: conditions} can be succinctly expressed as
\begin{equation}
    \mathrm{tr}\left(K^\dagger HKP\right)=0,
\end{equation}for any Pauli string $P\not\in\mathfrak{h}$ with a unitary $K$. Notice that this condition exactly corresponds to a non-zero global minimum of the cost function \eqref{opt2}. Therefore, given the same algebraic knowledge of the Cartan decomposition, our formulation in \eqref{opt2} can be equivalently interpreted as the optimization problem presented in \cite{kokcu2022fixed}, by expressing $K$ and $K^\dagger H K$ within the framework of Cartan decomposition. Nevertheless, we see that these Lie-diagonalization algorithms may be inefficient in general for the family of Hamiltonians in \cref{def: QD1} due to the existence of exponentially-large Lie algebras in \cref{ex: Hams} and \cref{ex: Hams1}.

\section{Conclusion}
In this work, we present classical optimization algorithms for Hamiltonian diagonalization and identify a family of Hamiltonians for which these algorithms remain efficient, while existing algorithms are not. Specifically, the proposed algorithms exhibit only polynomial per-iteration scaling for Hamiltonians whose eigen-decompositions consist of a diagonal matrix and a unitary operator that are both sparse in the Pauli basis, and whose optimization landscapes admit only global minima as stationary points. To this end, we formulate a cost function in terms of Pauli strings, which constitutes the off-diagonal terms of the Hamiltonian with an orthogonality constraint. We establish that any non-zero stationary point of the cost function corresponds to a global minimum, and eventually a diagonalization of the Hamiltonian. To find such a non-zero stationary point, we propose a simple optimization algorithm such that the gradient descent is implemented at each iteration with normalization. We show that this algorithm converges under a reasonable non-convex condition that is validated in numerical experiments. Our analysis also shows the role of a good initialization in enabling fast convergence.

Toward practical implementations, we also introduced a randomized-coordinate variant of our optimization algorithm that delivers a more efficient iteration scheme than its deterministic counterpart. By formulating a tailored update rule, this randomized approach achieves a quadratic scaling per-iteration cost with respect to the Hilbert-space dimension. Another key numerical observation is that a continuation strategy can accelerate convergence: the optimal parameters obtained for one Hamiltonian serve as an excellent initialization when diagonalizing a nearby Hamiltonian whose physical parameters differ only slightly. Consequently, our algorithms are well-suited for successively diagonalizing families of Hamiltonians that vary smoothly in a control parameter and for applications where rapid, “fast-forward” Hamiltonian simulation is crucial. A notable example is Hamiltonian simulation in the interaction picture \cite{low2018hamiltonian}, as well as certain ODE systems \cite{an2022theory}.

\section*{Code availabilty}
Code used for the current study are available at the following GitHub repository:

https://github.com/TKmath/Classical-algorithms-for-diagonalizing-quantum-Hamiltonians

\section*{Acknowledgement}
 This research was supported by Quantum Simulator Development Project for Materials Innovation through the National Research Foundation of Korea (NRF) funded by the Korean
government (Ministry of Science and ICT(MSIT))(RS-2023-NR119931). We used resources of the Center for Advanced Computation at Korea Institute for Advanced Study and the National Energy Research Scientific Computing Center (NERSC), a U.S. Department of Energy Office of Science User Facility operated under Contract No.DE-AC02-05CH11231. SC was supported by a KIAS Individual Grant (CG090601) at Korea Institute for Advanced Study. XL is supported by the NSF Grant DMS-2411120. TK is supported by a KIAS Individual Grant (CG096001) at Korea Institute for Advanced Study. 

\section*{Conflicts of interests}

All authors declare no financial or non-financial competing interest

\section*{Author contributions}
TK, SC, and XL conceptualized the project. Analyses and numerical experiments were led by TK and
discussed with SC, HP and XL. The manuscript was written by TK, and XL.

\appendix
\section{Optimization with only global minima and convergence analysis}
\label[appsec]{sec: Appendix A}
In this section, we present proofs for \cref{thm: F} and \cref{eq: convergence thm}. First, we prove \cref{thm: F}.
\begin{proof}
    Let $(\bm r_c,\bm\theta_c)$ be a non-zero stationary point, that is $\bm r_c\neq 0$. Denote $K_c=K(\bm r_c,\bm\theta_c)$. By \eqref{opt2} and \eqref{stograd}, it follows at the stationary point,
        \begin{equation}\label{eq: real1}
        \begin{split}
        0&=\frac{\partial F}{\partial r_j}\\
        &=\frac{\partial f}{\partial r_j}+2\sum_{P\in G_2}\phi_P(\bm r_c,\bm\theta_c)\frac{\phi_P(\bm r_c,\bm\theta_c)}{\partial r_j}\\
        &= \frac{\partial f}{\partial r_j}+4\sum_{P\in G_2}\phi_P(\bm r_c,\bm\theta_c)c_{j,P}r_{j_P}\exp(i(\theta_j-\theta_{j_P})).
        \end{split}
    \end{equation}Similar to the proof of \cref{thm: f}, by multiplying $r_j$ to \eqref{eq: real1} and summing over $j$'s, we arrive at
    \begin{equation}\label{eq: F=0}
    \begin{split}
        0&= 4f(\bm r_c,\bm \theta_c)+ 4\sum_jr_j\left(\sum_{P\in G_2}\phi_P(\bm r_c,\bm\theta_c)c_{j,P}r_{j_P}\exp(i(\theta_j-\theta_{j_P}))\right)\\
        &=4f(\bm r_c,\bm \theta_c)+ 4\sum_{P\in G_2}\phi_P(\bm r_c,\bm\theta_c)\left(\sum_jc_{j,P}r_jr_{j_P}\exp(i(\theta_j-\theta_{j_P}))\right)\\
        &=4f(\bm r_c,\bm\theta_c)+4\sum_{P\in G_2}\phi_P(\bm r_c,\bm\theta_c)^2=4F(\bm r_c,\bm \theta_c),
    \end{split}
    \end{equation}by definition of $\phi_P$ in \cref{lem: constraint}. Since $F(\bm r_c,\bm \theta_c)=0$, $(\bm r_c,\bm \theta_c)$ is a global minimum of the cost function in \eqref{opt2}.

    Now we prove that $K$ is unitary up to scaling. Noting that $\sum_{P\in G_2}\phi_P(\bm r_c,\bm\theta_c)^2=0$ from \eqref{eq: F=0}, we know from \cref{lem: constraint} that
    \begin{equation}
        K(\bm r_c,\bm\theta_c)^\dagger K(\bm r_c,\bm\theta_c) = \norm{\bm r_c}^2I.
    \end{equation}Additionally, since $f(\bm r_c,\bm\theta_c)=0$, it immediately follows that $K(\bm r_c,\bm\theta_c)^\dagger H K(\bm r_c,\bm\theta_c)$ is a diagonal matrix, which results that
    \begin{equation}
        H = \frac{1}{\norm{\bm r}^4}K(\bm r_c,\bm\theta_c)h(\bm r_c,\bm\theta_c)K(\bm r_c,\bm\theta_c)^\dagger,
    \end{equation}where $h(\bm r_c,\bm\theta_c):=K(\bm r_c,\bm\theta_c)^\dagger HK(\bm r_c,\bm\theta_c)$. This completes the proof.

\end{proof}
This proof shows that the optimization problem we consider has only global minima as stationary points, and especially any non-zero stationary point yields a diagonalization unitary $K$. Moreover, the following proof guarantees that \cref{alg: algorithm1} converges to a non-zero stationary point.
\begin{proof}
Since the function in \eqref{opt2} is a polynomial of degree $4$ with respect to $\bm r$ and $\bm \theta$ appear there as the phase factors that do not affect the magnitude of the function,  we can find a constant $L>0$ such that
\begin{equation}\label{eq: constants GL}
     \max_{\norm{\bm r(\bm x)}\leq 2}\norm{\nabla^2 F(\bm x)}\leq L,
\end{equation}in the neighborhood of radius $2$   defined in the $\bm r$ direction. To proceed, we derive a recursive inequality for \cref{alg: algorithm1}. Suppose that $\norm{\bm r(\bm x_t)}=1$, as \cref{alg: algorithm1} calls for normalization about $\bm r$ at each iteration. 

From the definition of $\bm y_t$ in \cref{alg: algorithm1} and the condition on $\bm x_t$, we notice that
\begin{equation}\label{eq: norm y}
\begin{split}
\norm{\bm r(\bm y_t)}^2 &= \sum_j((\bm r(\bm x_t))_j-a\frac{\partial F}{\partial r_j}(\bm x_t))^2=\norm{\bm r(\bm x_t)}^2-2a\bm r(\bm x_t)\cdot\nabla_r F(\bm x_t) + a^2\norm{\nabla_rF(\bm x_t)}^2\\
& = 1-8aF(\bm x_t) + a^2\norm{\nabla_rF(\bm x_t)}^2
\end{split}
\end{equation}where $\nabla_rF:=(\frac{\partial F}{\partial r_1},...,\frac{\partial F}{\partial r_N})^T$. The last equality follows from \cref{prop: grad-cost}. We derive a lower bound for $\norm{\bm r(\bm y_t)}$. By applying the Cauchy-Schwarz inequality to \cref{prop: grad-cost} in \eqref{eq: norm y}, we have
\begin{equation}\label{eq: norm y lb}
    \norm{\bm r(\bm y_t)}^2\geq 1-8aF(\bm x_t) + 16a^2F(\bm x_t)^2\geq 1-8aF(\bm x_t).
\end{equation}Then for any $a\leq \frac{2}{L}$, we achieve that
\begin{equation}\label{eq: y-x}
\begin{split}
    F(\bm y_{t})&\leq F(\bm x_t) - a\norm{\nabla F(\bm x_t)}^2+\frac{La^2}{2}\norm{\nabla F(\bm x_t)}^2 \\
    &\leq F(\bm x_t) - \frac{a}{2}\norm{\nabla F(\bm x_t)}^2\\
    &\leq (1- 2\mu aF(\bm x_t)^{\alpha-1})F(\bm x_t).
\end{split}
\end{equation}The last inequality holds by the assumption \eqref{eq: kl condition}. By \eqref{eq: norm y lb} and \eqref{eq: y-x}, we observe that 
\begin{equation}\label{eq: reineq}
\begin{split}
    F(\bm x_{t+1})=\frac{1}{\norm{\bm r(\bm y_t)}^4}F(\bm y_t)&\leq \frac{1-2\mu aF(\bm x_t)^{\alpha-1}}{\left(1-8aF(\bm x_t)\right)^2}F(\bm x_t)\\
    &\leq \frac{1-2\mu aF(\bm x_t)^{\alpha-1}}{1-16aF(\bm x_t)}F(\bm x_t)\\
    &\leq \left(1-(2\mu-16F(\bm x_t)^{2-\alpha})aF(\bm x_t)^{\alpha-1}\right)F(\bm x_t)\\
    &\leq \left(1-\mu aF(\bm x_t)^{\alpha-1}\right)F(\bm x_t),
\end{split}
\end{equation}when $F(\bm x_t)\leq \left(\frac{\mu}{16}\right)^{\frac{1}{2-\alpha}}$. To find the number of iterations $T$ to achieve $\epsilon$ precision, we solve the following inequality,
\begin{equation}
    \epsilon\leq \left(1-\mu a\epsilon^{\alpha-1}\right)^TF(\bm x_0).
\end{equation}
Therefore, for a reasonable initial guess $\bm x_0$, \cref{alg: algorithm1} requires the number of iterations scaling as 
\begin{equation}
    T=\mathcal{O}\left(\frac{L}{\mu  \epsilon^{\alpha-1}}\log\frac{F(\bm x_0)}{\epsilon}\right).
\end{equation}

\end{proof}

\section{Sensitivity analysis for numerical solutions}
\label[appsec]{sec: Appendix B}
In this section, we prove \cref{lem: perturb} and  \cref{thm: perturb}. First, we prove \cref{lem: perturb}.
\begin{proof}
For ease of notation in analysis, we omit the parameters $\bm r, \bm \theta$ unless stated otherwise. Recalling the set $G_1$ in \eqref{opt} and denoting $G_3=\{\otimes_{i=1}^n\sigma^{(i)} :\sigma^{(i)}\in\{I,X,Y,Z\} \text{ and }\forall i\in[n],\;  \sigma^{(i)}\in\{I,Z\}\}$, we define 
\begin{equation}\label{eq: h0 Delta}
h_0= \frac{1}{2^n}\sum_{P\in G_3} \tr(K^\dagger H KP)P,\quad \Delta = \frac{1}{2^n}\sum_{P\in G_1} \tr(K^\dagger H KP)P,
\end{equation}which sum to the $K^\dagger HK$, thus verifying \cref{h0Delta}. 
By definition \eqref{opt}, we observe that
\begin{equation}
    \norm{\Delta}_F^2=\frac{1}{2^n}\sum_{P\in G_1} \tr(K^\dagger H KP)^2=\frac{f}{2^n}\leq \frac{F}{2^n},
\end{equation}since $f\leq F$ by \eqref{opt2}. This proves the first statement of \cref{lem: perturb}. For the second statement, by our assumptions, we have from the proof of \cref{lem: constraint} that,
\begin{equation}\label{eq: KK close to I}
    \norm{K^\dagger K-I}_F^2 =\norm{\sum_{P\in G_2}\phi_P(\bm r,\bm \theta)P}_F^2=2^n\sum_{P\in G_2}\phi_P(\bm r,\bm \theta)^2\leq \epsilon\leq \frac{1}{4}.
\end{equation}Note that the inverse of $K$ exists, and 
    \begin{equation}
        \norm{K^{-1}}_2^2=\norm{(K^\dagger K)^{-1}}_2=\norm{(K^\dagger K-I+I)^{-1}}_2\leq \frac{1}{1-\sqrt{\epsilon}}\leq 2.
    \end{equation}From this, we have the error in the spectral norm, 
\begin{equation}
    \begin{split}
        \norm{H-\widetilde{H}}_2&=\norm{(K^\dagger)^{-1}K^\dagger\left(H-\widetilde{H}\right)KK^{-1}}_2\\
        &=\norm{(K^\dagger)^{-1}\left(h_0+\Delta-K^\dagger \widetilde{H}K\right)K^{-1}}_2\\
        &=\norm{(K^\dagger)^{-1}\left(h_0+\Delta-K^\dagger Kh_0K^\dagger K\right)K^{-1}}_2\\
        & \leq 2\norm{h_0+\Delta-K^\dagger Kh_0K^\dagger K}_2.
    \end{split}
    \end{equation} The term in the last line is bounded above as
    \begin{equation}
        \begin{split}
            &\norm{h_0+\Delta -K^\dagger Kh_0K^\dagger K}_F\\
            &=\norm{h_0+\Delta -(K^\dagger K-I+I)h_0(K^\dagger K-I+I)}_F\\
            &=\norm{\Delta -(K^\dagger K-I)h_0(K^\dagger K-I)-(K^\dagger K-I)h_0-h_0(K^\dagger K-I)}_F\\
            &\leq \frac{F}{2^n}+3\norm{h_0}_F\sqrt{\epsilon}.
        \end{split}
    \end{equation}
    
    Finally, noting that
    \begin{equation}
\norm{h_0}_F^2=\frac{1}{2^n}\sum_{p\in\mathfrak{h}}\tr(K^\dagger HKp)^2\leq \frac{1}{2^n}\norm{K^\dagger HK}_F^2\leq \frac{1}{2^n}\norm{K^\dagger K}_F^2\norm{H}_F^2\leq (1+\sqrt{F})^2\norm{H}_F^2,
    \end{equation}with the above results, we obtain 
    \begin{equation}
        \norm{H-\widetilde{H}}_2\leq \frac{F}{2^{n-1}}+6(1+\sqrt{F})\norm{H}_F\sqrt{\epsilon},
    \end{equation}which proves the second statement of \cref{lem: perturb}.

\end{proof}
This completes the proof of \cref{lem: perturb}. Using the result, we prove \cref{thm: perturb} that gives an approximation error bound of diagonalization of Hamiltonian for a numerical solution obtained from our algorithm. 
\begin{proof}From \cref{lem: perturb},  we notice that
$\norm{H-\widetilde{H}}_2\leq \epsilon'$. Next, we refer to the proof of \cite[Lemma 5]{ko2024quantum} from which we can deduce that 
 \begin{equation}
     \bra{\widetilde{\psi}_j}P_{\lambda_k(H)}\ket{\widetilde{\psi}_j}\geq 1-(C+1)\epsilon',
 \end{equation}where $\widetilde{\psi}_j$ denotes any eigenvector of $P_{k,\widetilde{H}}$ and the constant $C$ depends on the eigenvalues of $H$ by considering a fixed polynomial approximation in the Stone-Weierstrass theorem to the indicator function in \cite[eq (38)]{ko2024quantum}.  By summing this inequality over $j$'s, we have
 \begin{equation}
     \text{tr}\left(P_{k,\widetilde{H}}P_{\lambda_k(H)}\right)\geq \text{rank}(P_{\lambda_k(H)})\left(1-(C+1)\epsilon'\right),
 \end{equation}since $\text{rank}(P_{k,\widetilde{H}}) =\text{rank}(P_{\lambda_k(H)})$, and therefore
 \begin{equation}
 \begin{split}
     &\norm{P_{k,\widetilde{H}}-P_{\lambda_k(H)}}_F^2\leq 2\text{rank}(P_{\lambda_k(H)})-2\text{tr}\left(P_{k,\widetilde{H}}P_{\lambda_k(H)}\right)\\
     &\leq 2\text{rank}(P_{\lambda_k(H)})(C+1)\epsilon'\leq 2\max_k\text{rank}(P_{\lambda_k(H)})(C+1)\epsilon'.
 \end{split}
 \end{equation}Therefore,  the statement is proven.
\end{proof}

\section{Quantum diagonalizable Hamiltonians whose Lie algebras are exponentially large}
\label[appsec]{sec: Appendix C}
In this section, we prove that the Lie algebras of Hamiltonians in \cref{ex: Hams} are exponentially sized.
\begin{proof}
We consider the trivial case where $\prod_{m=1}^{\lfloor \log n\rfloor}U_m=I$. If $\mathfrak{g}(H)=\text{su}(2^n)$ in this case, then any Hamiltonian in \cref{ex: Hams} would satisfy the same property.

To proceed, we make use of the result in \cite{smith2024optimally}, which states that the set,
\begin{equation}\label{eq: su-generator}
    \{Z_1, Z_2, X_1, X_2, Z_1Z_2\}\cup\{X_2Y_3...Z_j, Z_2Y_3...X_j|3\leq j\leq n\}, 
\end{equation}generates all Pauli strings in the sense of dynamical Lie algebra. Noticing that the second factor of $U$ consists of mutually-anticommuting Pauli strings, we see that the constraint that $\sum_{j=0}^nc_j^2=1$ and $c_j$'s are real suffices to ensure that $U$ is a unitary, a property of the anticommuting Pauli strings \cite{izmaylov2019unitary,ryabinkin2023efficient}. Now we prove the statement by showing that all Pauli strings in the set \eqref{eq: su-generator} are contained in the Hamiltonian $H$ and its dynamical Lie algebra.

First, we observe that $U$ contains the identity $I$ as its component, since $i\sin\theta Z_2\times c_2Z_2 = ic_2\sin\theta I$ and $c_2,\theta\neq 0$. This implies that $H$ contains the set $\{X_2Y_3..Z_j|3\leq j \leq n\}$, since the elements of Pauli string are chosen from $U\times I\times ic_2\sin\theta I$ in the multiplication $UDU^\dagger$. Similarly, we can find that $Z_1Y_2\times Y_2\times I=Z_1$, $X_1Y_2\times Y_2\times I=X_1$, $Z_2\times I\times I=Z_2$, $Z_2\times Y_2\times I=X_2$, and $Z_1X_2\times Y_2\times I=Z_1Z_2$. Notice that  $Z_1X_2$ is in $U$ due to the multiplication, $i\sin\theta Z_2\times c_1Z_1Y_2$. Lastly, we are left to show that the Pauli strings in $\{Z_2Y_3...X_j|3\leq j\leq n\}$ are contained in $\mathfrak{g}(H)$. We observe that since $I\times Y_j\times I=Y_j$, $Y_j$'s are in $H$. Since  $H$ contains the set $\{X_2Y_3..Z_j|3\leq j \leq n\}$ as shown above, and  $[[X_2Y_3...Z_j,Y_2],Y_j]$ equals $Z_2Y_3...X_j$ up to a scaling factor, it happens that $\mathfrak{g}(H)$ also contains the set $\{Z_2Y_3...X_j|3\leq j\leq n\}$. Therefore, $\mathfrak{g}(H)$ contains the generating set \eqref{eq: su-generator} and all $n$-qubit Pauli strings.

\end{proof}

\bibliographystyle{plain}
\bibliography{ref}

\end{document}